\documentclass[11pt]{article}
\usepackage{amsfonts,amsmath,amsthm,amssymb,dsfont, etex,bm,mathtools}
\usepackage[colorlinks=true,linkcolor=blue,citecolor=blue]{hyperref}
\usepackage{paralist}
\usepackage{enumerate}
\usepackage{framed}
\usepackage{mdframed}

\usepackage{mathpazo}

\usepackage{tikz}
\usetikzlibrary{arrows,automata}
\usepackage{comment}
\usepackage{graphicx,float,wrapfig}
\usepackage{mathrsfs}
\usepackage[usenames,dvipsnames]{pstricks}
\usepackage{epsfig}
\usepackage{pst-grad} 
\usepackage{pst-plot} 
\usepackage{mathdots}
\usepackage[linesnumbered,boxed,ruled,vlined]{algorithm2e}

\usepackage{bbm}
\usepackage{url}
\usepackage{tabu}
\usepackage{xfrac}
\usepackage[normalem]{ulem}
\usepackage{bm}
\usepackage{mleftright}

\usepackage{hyperref}  
\hypersetup{colorlinks=true,citecolor=blue,linkcolor=blue}

\usepackage[capitalize]{cleveref}

\usepackage[breakable]{tcolorbox}
\newtcolorbox{construction}[2][]
{
	breakable,
	colframe = gray!50,
	colback  = gray!10,
	coltitle = gray!10!black,
	before skip = 10pt,
	after skip = 10pt,
	title    = \textbf{#2},
	#1,
}
\newtcolorbox{graphview}[2][]
{
	breakable,
	colframe = black!30,
	colback  = black!0,
	coltitle = gray!10!black,
	before skip = 10pt,
	after skip = 10pt,
	title    = \textbf{#2},
	#1,
}

\newenvironment{reminder}[1]{\bigskip
	\noindent {\bf Reminder of #1.  }\em}{\smallskip}

\renewcommand{\epsilon}{\varepsilon}
\newcommand{\Ex}{\mathop{\mathbb{E}}}

\newcommand{\poly}{\mathop{\mathrm{poly}}}
\newcommand{\polylog}{\mathop{\mathrm{polylog}}}
\newcommand{\eps}{\varepsilon}

\newcommand{\N}{\mathbb{N}}

\newcommand{\bits}{\{0,1\}}

\newcommand{\WT}{\widetilde}

\newcommand{\zeroTon}[1]{\{0,1,\dotsc,#1\}}

\newcommand{\QMA}{\mathsf{QMA}\xspace}

\newcommand{\NP}{\mathsf{NP}\xspace}

\def \ro {\text{r.o.-}}

\newcommand{\cro}[1]{$c$\text{-}\ro}

\newcommand{\PCP}{\mathsf{PCP}\xspace}

\newcommand{\Enc}{\mathsf{Enc}}
\newcommand{\Dec}{\mathsf{Dec}}

\newcommand{\negl}{\mathsf{negl}}

\newcommand{\ie}{\textit{i}.\textit{e}.\@\xspace}
\newcommand{\eg}{\textit{e}.\textit{g}.\@\xspace}

\def\poly{\mathrm{poly}}
\def\polylog{\mathrm{polylog}}

\def\caC{\mathcal{C}}
\def\caH{\mathcal{H}}
\def\caP{\mathcal{P}}
\def\caS{\mathcal{S}}

\def\caD{\mathcal{D}}
\def\caG{\mathcal{G}}

\usepackage[draft,inline,marginclue,index]{fixme}  

\fxsetup{theme=color,mode=multiuser}
\FXRegisterAuthor{clj}{aclj}{\colorbox{black!20!green}{\color{white}Lijie}}  

\usepackage[margin=1in]{geometry}
\newcommand{\makeName}[1]{%
	\expandafter\newcommand\csname#1\endcsname{\mathsf{#1}}}

\usepackage{aliascnt} 

\usepackage{mathpazo}
\newtheorem{theorem}{Theorem}[section]
\newtheorem*{theorem*}{Theorem}

\newaliascnt{definition}{theorem}
\theoremstyle{definition}
\newtheorem{definition}[definition]{Definition}
\aliascntresetthe{definition}

\newtheorem*{definition*}{Definition}

\theoremstyle{plain}
\newaliascnt{lemma}{theorem}
\newtheorem{lemma}[lemma]{Lemma}
\aliascntresetthe{lemma}

\newtheorem*{lemma*}{Lemma}

\newaliascnt{claim}{theorem}
\newtheorem{claim}[claim]{Claim}
\aliascntresetthe{claim}

\newtheorem*{claim*}{Claim}

\newaliascnt{fact}{theorem}

\aliascntresetthe{fact}

\newtheorem*{fact*}{Fact}

\newaliascnt{observation}{theorem}

\aliascntresetthe{observation}

\newtheorem*{observation*}{Observation}

\newaliascnt{conjecture}{theorem}
\newtheorem{conjecture}[conjecture]{Conjecture}
\aliascntresetthe{conjecture}

\newtheorem*{conjecture*}{Conjecture}

\newaliascnt{corollary}{theorem}
\newtheorem{corollary}[corollary]{Corollary}
\aliascntresetthe{corollary}

\newtheorem*{corollary*}{Corollary}

\newaliascnt{remark}{theorem}
\newtheorem{remark}[remark]{Remark}
\aliascntresetthe{remark}

\newtheorem*{remark*}{Remark}

\newaliascnt{proposition}{theorem}

\aliascntresetthe{proposition}

\newtheorem*{proposition*}{Proposition}

\makeatletter
\patchcmd{\ALG@step}{\addtocounter{ALG@line}{1}}{\refstepcounter{ALG@line}}{}{}
\newcommand{\ALG@lineautorefname}{Line}
\makeatother

\usepackage{fancybox}
\usepackage{tikz}
\usetikzlibrary{shapes,backgrounds}
\usepackage{tikzsymbols}

\usetikzlibrary{decorations.pathmorphing}
\usetikzlibrary{decorations.pathreplacing}

\tikzset{snake it/.style={decorate, decoration=snake}}

\newcommand{\Pisuc}{\protocol_{\sf succinct}}
\newcommand{\LH}{\mathsf{LH}}
\newcommand{\QPCP}{\mathsf{QPCP}}

\newcommand{\lmin}{\lambda_{\sf min}}
\newcommand{\tr}{\mathrm{tr}}
\newcommand{\Tr}{\mathrm{Tr}}

\newcommand{\spz}[1]{|#1\rangle}

\newcommand{\rpz}[1]{\langle #1 |}

\newcommand{\yes}{\mathsf{yes}}
\newcommand{\no}{\mathsf{no}}
\newcommand{\SimQMA}{\mathsf{SimQMA}}

\newcommand{\ayes}{L_{\sf yes}}
\newcommand{\ano}{L_{\sf no}}

\newcommand{\protocol}{\Pi}
\newcommand{\msg}{\mathsf{msg}}

\newcommand{\CNOT}{\mathsf{CNOT}}

\newcommand{\EncprojLH}{\mathsf{ELH}}

\newcommand{\state}{\mathsf{state}}
\newcommand{\witness}{\mathsf{witness}}
\newcommand{\ancilla}{\mathsf{ancilla}}

\newcommand{\data}{\mathsf{data}}
\newcommand{\out}{\mathsf{out}}
\newcommand{\clock}{\mathsf{clock}}
\newcommand{\idx}{\mathsf{idx}}
\newcommand{\midx}{\mathsf{midx}}

\newcommand{\echk}{\mathsf{Echk}}
\newcommand{\edata}{\mathsf{Edata}}
\newcommand{\eout}{\mathsf{Eout}}
\newcommand{\eidx}{\mathsf{Eidx}}
\newcommand{\emidx}{\mathsf{Emidx}}

\newcommand{\Hstate}{\caH_{\state}}

\newcommand{\Hclock}{\caH_{\clock}}

\newcommand{\proj}[1]{\spz{#1}\rpz{#1}}
\newcommand{\innerprod}[2]{\rpz{#1}#2\rangle}

\newcommand{\Hamprop}{H^{\sf prop}}
\newcommand{\Haminit}{H^{\sf in}}
\newcommand{\Hamout}{H^{\sf out}}
\newcommand{\Hamhis}{H^{\sf history}}
\newcommand{\Hamcheck}{H^{\sf check}}
\newcommand{\Hamstab}{H^{\sf stab}}

\newcommand{\reg}{\mathsf{reg}}

\newcommand{\caV}{\mathcal{V}}

\newcommand{\Venc}{\caV^{\mathsf{enc}}}
\newcommand{\VHenc}{\caV^{\mathsf{enc}\text{-}\mathsf{H}}}

\newcommand{\Utestenc}{U_{\sf etest}}
\newcommand{\UcheckE}{U_{\sf checkE}}
\newcommand{\supp}{\mathrm{supp}}

\newcommand{\other}{\mathsf{other}}

\newcommand{\Pizk}{\Pi_{\sf zk}}

\newcommand{\sQ}{\mathsf{Q}}

\newcommand{\syndrome}{\mathsf{syndrome}}
\newcommand{\Hsyndrome}{\caH_{\syndrome}}
\newcommand{\Chk}{\mathsf{Chk}}

\newcommand{\History}{\mathsf{History}}

\newcommand{\footremember}[2]{%
    \footnote{#2}
    \newcounter{#1}
    \setcounter{#1}{\value{footnote}}
}

\newcolumntype{C}[1]{%
 >{\vbox to 4ex\bgroup\vfill\centering}%
 p{#1}%
 <{\egroup}} 

\def\ShowAuthNotes{1}
\ifnum\ShowAuthNotes=1
\newcommand{\authnote}[2]{\ \\ \textcolor{red}{\parbox{0.9\linewidth}{[{\footnotesize {\bf #1:} { {#2}}}]}}\newline}
\else
\newcommand{\authnote}[2]{}
\fi
\newcommand{\lnote}[1]{\authnote{Lijie}{#1}}

\ifnum\ShowAuthNotes=1
\newcommand{\ramis}[1]{\textcolor{magenta}{[RM: \emph{#1}]}}
\else
\newcommand{\ramis}[1]{}
\fi

\title{Making Quantum Local Verifiers Simulable with Potential Applications to Zero-Knowledge} 
\date{\small \today} 
\author{Lijie Chen\footremember{UC Berkeley}{Miller Institute for Basic Research in Science, University of California Berkeley, Berkeley, CA, 94720, U.S.A.}
\and Ramis Movassagh\footremember{ibmcambridge}{IBM Quantum Research, MIT-IBM Watson AI Research Lab, Cambridge, MA, 02142, U.S.A.}}

\newcommand{\QHROM}{\mathsf{QHROM}}
\newcommand{\Haar}{\mathbb{U}}
\newcommand{\getsR}{\in_{\sf R}}
\newcommand{\unary}{\mathsf{unary}}
\newcommand{\QROM}{\mathsf{QROM}}

\newcommand{\commit}{\mathsf{commit}}
\newcommand{\decommit}{\mathsf{decommit}}

\newcommand{\LocalQMA}{\mathsf{LocalQMA}}
\newcommand{\otp}{\mathsf{otp}}
\newcommand{\eotp}{\mathsf{Eotp}}
\newcommand{\CZ}{\mathsf{C}\text{-}\mathsf{Z}}

\newcommand{\eT}{\mathsf{Emagic}}
\newcommand{\eancC}{\mathsf{Eanc}}
\newcommand{\Had}{\mathsf{H}}

\newcommand{\nwit}{n_{\mathsf{wit}}}
\newcommand{\nanc}{n_{\mathsf{anc}}}

\newcommand{\highlightuline}[1]{\medskip\noindent\uline{\textbf{#1}.}}

\begin{document}
	
	\maketitle

\begin{abstract}
Recently Chen and Movassagh proposed the quantum Merkle tree~\cite{ChenM22}, which is a quantum analogue of the well-known classical Merkle tree. It gives a succinct verification protocol for quantum state commitment. Although they only proved security against semi-honest provers, they conjectured its general security.

	
Using the proposed quantum Merkle tree,~\cite{ChenM22} gave a quantum analogue of Kilian's succinct argument for $\NP$, which is based on probabilistically checkable proofs (PCPs). A nice feature of Kilian's argument is that it can be extended to a zero-knowledge succinct argument for $\NP$, if the underlying PCP is zero-knowledge. Hence, a natural question is whether one can also make the quantum succinct argument by Chen and Movassagh zero-knowledge as well.
	
This work makes progress on this problem. We generalize the recent result of Broadbent and Grilo~\cite{BroadbentG20} to show that any local quantum verifier can be made \emph{simulable} with a minor reduction in completeness and soundness. Roughly speaking, a local quantum verifier is simulable if in the yes case, the local views of the verifier can be computed without knowing the actual quantum proof; it can be seen as the quantum analogue of the classical zero-knowledge PCPs. Hence we conjecture that applying the proposed succinct quantum argument of~\cite{ChenM22} to a simulable local verifier is indeed zero-knowledge.
	
\end{abstract}

\section{Introduction}

Chen and Movassagh~\cite{ChenM22} recently proposed the quantum Merkle tree, together with a candidate quantum succinct argument construction for the Gap-$k$-$\LH$ problem.

\begin{definition}
	(Gap-$k$-Local Hamiltonian Problem) Given  $\alpha,\beta$ with $0<\alpha<\beta\le 1$ and a $k$-local Hamiltonian with $m$ local terms $\{ H_i \}_{i \in [m]}$ such that $0 \le  H_i \le I$, decide whether $\lmin(\sum_{i=1}^{m}H_i)$ is at most $\alpha m$ or at least $\beta m$. Below we abbreviate this problem by $(\alpha,\beta)\text{-}k\text{-}\LH$.
\end{definition}

\paragraph*{Motivation: making the Chen-Movassagh construction zero-knowledge?} Assuming the quantum PCP conjecture, which says the above problem is $\QMA$-complete for some constants $\alpha < \beta$ (see~\autoref{sec:QPCP}), the proposed quantum succinct argument applies to all of $\QMA$. The construction from~\cite{ChenM22} can be seen as a quantum analogue of the well-known succinct argument for $\NP$ by Kilian~\cite{Kilian92}, which is based on probabilistically checkable proofs (PCPs). An important feature of the construction of~\cite{Kilian92} is that if the underlying $\PCP$ is zero-knowledge~\cite{DworkFKNS92,KilianPT97}, then so is the succinct argument.\footnote{It is zero-knowledge even in the quantum random oracle model ($\QROM$); see~\cite{ChiesaMS19}.}

A natural question, which is asked explicitly in~\cite{ChenM22}, is whether one can make their candidate quantum succinct argument zero-knowledge as well. Recall that a language $L \in \NP$ admits a succinct classical zero-knowledge PCP, if there is a \emph{local verifier $V_L$} that queries $\polylog(n)$ bits of a given proof $\pi$ such that (1) if $x \in L$, given a corresponding witness $w$, a poly-time algorithm can sample a proof $\pi$ from a distribution $\caD_{x,w}$, such that $V_L$ accepts $\pi$ with probability $1 - \negl(n)$, yet the local views of $V_L$ can be simulated \emph{without knowing} the proof $\pi$ and (2) if $x \notin L$, then $V_L$ accepts any proof $\pi$ with probability at most $\negl(n)$; see~\cite[Section~3.4]{ChiesaMS19} for a formal definition. To summarize, $L$ admits a zero-knowledge PCP if there is a \emph{locally simulable verifier} $V_L$ for $L$ with $1 - \negl(n)$ completeness and $\negl(n)$ soundness. 

Hence, the first step of making their construction zero-knowledge, is to obtain a quantum analogue of a simulable local verifier for the Gap-$k$-$\LH$ problem. Recently, \cite{BroadbentG20} proved that every language $L \in \QMA$ has a \emph{locally simulable verifier} $V_L$, such that (1) $V_L$ only queries $O(1)$ qubits in the witness; (2) given a yes instance $x \in L_{\sf yes}$, there exists a witness $\sigma$ that makes $V_L$ accepts with probability $1-n^{-\omega(1)}$, such that $V_L$'s local views\footnote{If $V_L$ decides to query an $O(1)$-size subset $S$, then her local view is $\Tr_{\bar{S}}[\sigma]$.} can be simulated \emph{without knowing} $\sigma$; (3) for every no instance $x \in L_{\no}$, $V_L$ rejects with probability at least $1/n^{c}$, for some constant $c \ge 3$. In other words, they proved $\QMA \subseteq O(1)\text{-}\SimQMA_{1-n^{-c},1-n^{-\omega(1)}}$ (see~\autoref{defi:simQMA} for the formal definition).~\cite{BroadbentG20}'s construction works for all of $\QMA$ but it has a very poor soundness of $1 - 1/n^{c}$. One may hope that it can achieve the require $\negl(n)$ soundness if we apply it only to the Gap-$k$-$\LH$ problem instead of $\QMA$ in general.

\paragraph*{Our result: locally simulable verifiers with minor loss in parameters.} \cite{BroadbentG20}'s result builds on the techniques of simulable codes/proofs from~\cite{GriloSY19}. It gives a transformation from a general quantum verifier $V$ for $L \in \QMA$ to a locally simulable verifier $V_L$ for $L$. However, due to the use of the Feynman-Kitaev clock construction~\cite{KSV02}, the rejection probability of $V_L$ in the no case is \emph{only inverse polynomial}, even if we apply their construction to the Gap-$k$-$\LH$ problem instead of $\QMA$ in general. Thus, it is not immediately clear how to use their construction of locally simulable verifier for our purpose.

In this paper, we generalize the result of~\cite{BroadbentG20} and show that if one starts from a local verifier $V$ for $L$ with a constant completeness/soundness gap (\eg, the natural local verifier for the Gap-$k$-$\LH$ problem), we can obtain a \emph{locally simulable verifier} with a constant completeness/soundness gap, at the cost of increasing the adaptivity by $1$.

Formally, we use $(k,\gamma)\text{-}\LocalQMA_{s,c}$ to denote all languages with a local verifier $V$ such that (1) $V$ queries at most $k$ qubits non-adaptively, and then applies a $\gamma$-size quantum circuit on the queried qubits to decide to accept or not, and (2) $V$ has completeness $c$ and soundness $s$ (see~\autoref{defi:localQMA} for the formal definition). We also use $(k,\ell)\text{-}\SimQMA$ to denote all languages with a locally simulable verifier that can ask $\ell$ rounds of queries (see~\autoref{defi:simQMA}).

\begin{theorem}\label{theo:LocalQMA-in-SimQMA}
	For every $k,\gamma \in \N$, $0 < \beta < 1$, and negligible function $\alpha$, there are $s \in (0,1)$ and $c \colon \N \to [0,1]$ such that $1 - c(n) \le \negl(n)$ and the following holds
	\[
	(k,\gamma)\text{-}\LocalQMA_{\beta,1-\alpha(n)} \subseteq (O(\log n),2)\text{-}\SimQMA_{s,c(n)}.
	\]
\end{theorem}

\highlightuline{Techniques} We remark that it is not clear how to directly adapt~\cite{BroadbentG20}'s transformation to prove~\autoref{theo:LocalQMA-in-SimQMA}. Indeed, we have to apply the Feynman-Kitaev clock construction together with the simulable codes to a sequence of \emph{non-local unitaries} (see~\autoref{sec:decompose-unitary}). The crucial observation we made here is that the local simulation of simulable codes in~\cite{BroadbentG20,GriloSY19} indeed works for every subset $S$ that has small intersections with every encoding block (see~\autoref{lemma:n-copy-sim}). See~\autoref{sec:locl-sim} for a proof of~\autoref{theo:LocalQMA-in-SimQMA}.

\paragraph*{Application to quantum zero-knowledge PCP and zero-knowledge succinct quantum arguments.}

The quantum PCP conjecture can be equivalently stated as follows\footnote{For technical reason, we will need a stronger version with $1 - \negl(n)$ completeness.}:

\begin{conjecture}[Quantum PCP conjecture with almost perfect completeness, an equivalent formulation]\label{conj:quantum-pcp-intro}
	There are constants $k,\gamma \in \N$ and $s \in (0,1)$ such that $\QMA \subseteq (k,\gamma)\text{-}\LocalQMA_{s,1 - \negl(n)}$.
\end{conjecture}

The following is an immediate corollary of~\autoref{theo:LocalQMA-in-SimQMA}.
\begin{corollary}
	Assuming~\autoref{conj:quantum-pcp-intro} holds, $\QMA \subseteq  (O(\log n),2)\text{-}\SimQMA_{s,1 - \negl(n)}$ for some constants $s \in (0,1)$.
\end{corollary}

The above corollary can be interpreted as quantum PCPs implies quantum \emph{zero-knowledge} PCPs. 

Finally, applying the candidate quantum succinct argument from~\cite{ChenM22} to $(O(\log n),2)$-$\SimQMA$, we obtain a candidate construction of quantum succinct zero-knowledge argument in $\QHROM$; see~\autoref{sec:zk-suc-cand} for more details.

\subsection{Related Works}

\paragraph*{Zero-knowledge protocols for $\QMA$.} Zero-knowledge protocols for $\QMA$ were recently studied in various works. Most of these works need assumptions stronger than the existence of post-quantum OWFs, such as the quantum hardness of LWE (QLWE), the existence of quantum-secure Fully-Homomorphic Encryption (QFHE), and the existence of quantum-secure indistinguishability obfuscation (QiO) (or combinations of them)~\cite{AlagicCGH20,ColadangeloVZ20,BitanskyS20,Bartusek21,Shmueli20,ChiaCY20,ChardouvelisM21,BartusekKLMMVVY22}. We refer the interested readers to these papers for more details.

\paragraph*{Succinct arguments against quantum adversaries.}~\cite{ChiesaMS19} proved that the succinct arguments for $\NP$ in~\cite{Kilian92,Micali00} are also secure in $\QROM$. Recently, it was proved that Kilian's four-message succinct argument for $\NP$~\cite{Kilian92} is secure in the standard model against quantum adversaries under QLWE~\cite{Chiesa21}. In~\cite{ChiaCY20}, under QLWE, QFHE, QiO, and some other assumptions, a succinct argument for $\mathsf{QTIME}(T)$ was constructed, in which the communication between the prover and the verifier is \emph{classical}. Later, in~\cite{BartusekKLMMVVY22}, the same succicnt argument was constructed under only QLWE and QiO.


\section*{Organization}

In~\autoref{sec:prelim}, we provide the necessary preliminaries for this paper. In~\autoref{sec:locl-sim}, we show that every local verifier can be made simulable and prove~\autoref{theo:LocalQMA-in-SimQMA}. In~\autoref{sec:zk-suc-cand}, we present our candidate zero-knowledge succinct argument for the Gap-$k$-$\LH$ problem.
	
\section{Preliminaries}\label{sec:prelim}

\subsection{Notation}

We always denote by $\log$ the logarithm in base $2$. We denote by $[n]$ the set of integers $\{1,2,\dots,n\}$.  Let $\reg$ be a register of $n$ qubits. For each $i \in [n]$, $\reg(i)$ denotes the $i$-th qubit in $\reg$, and $\reg[\ell,r]$ denotes the qubits from $\reg(\ell)$ to $\reg(r)$. The corresponding Hilbert space is denoted by $\caH_{\reg}$. For $k$ pairwise-disjoint sets $S_1,\dotsc,S_k$, we use $\bigsqcup_{i \in [k]} S_i$ to denote their union. We say a function $\alpha \colon \N \to [0,1]$ satisfies $\alpha(n) \le \negl(n)$ (\ie, $\alpha$ is \emph{negligible}), if for all constant $k \ge 1$, $\lim_{n \to \infty} \alpha(n) \cdot n^k = 0$ (\ie, $\alpha(n) = o(1/n^k)$ for every $k \in \N$).

For a quantum state $\sigma$ on $n$ qubits and a subset $S \subseteq [n]$, $\sigma_{S} \equiv \Tr_{[n] \setminus S}[\sigma]$ is the reduced density matrix. For a quantum state $\spz{\psi} \in \caH_{\reg}$, for simplicity we sometimes use $\psi$ to denote the corresponding density matrix $\psi=\proj{\psi}$. Given a unitary sequence $U_{1},\dotsc,U_{T}$, we write $U_{[\ell,r]}$ to denote the product $U_{r}U_{r-1}\cdots U_{\ell}$ for ease of notation.

For two quantum states $\sigma$ and $\rho$, we use $\| \sigma - \rho \|_1$ to denote their trace distance. We also write $x \getsR A$ to mean that $x$ is drawn from the set $A$ uniformly at random.

\subsection{History States and History Hamiltonians}

We will use the Feynman-Kitaev clock construction~\cite{KSV02} in this paper. Below we introduce the definitions and some important related results.

For $T \in \N$ and $t \in \zeroTon{T}$, we define the $T$-qubit state $\unary(t,T)$ as
\[
\spz{\unary(t,T)} =  \spz{1^{t}} \otimes \spz{0^{(T-t)}}.
\]

For simplicity, we often write it as $\unary(t)$ when $T$ is clear from the context.

\begin{definition}[History states]
	Let $U$ be a unitary that acts on $\Hstate = \caH_{\witness} \otimes \caH_{\ancilla}$, with $n_1$ qubits in $\witness$ and $n_2$ qubits in $\ancilla$. Let  $U_1,\dotsc,U_{T}$ be a sequence of unitaries such that $U = U_T \dotsc U_1$ (note that these $U_i$'s may not be local). Then a state $\spz{\Psi} \in \caH_{\clock} \otimes \caH_{\state}$ is a history state, if
	\[
	\spz{\Psi} = \frac{1}{\sqrt{T+1}} \sum_{t=0}^{T} \spz{\unary(t)}_{\clock} \otimes \spz{\psi_t}_{\state},
	\]
	where $\spz{\psi_t} = U_{[1,t]} \spz{\psi_{0}}$, and $\spz{\psi_0} = \spz{\phi}_{\witness} \otimes \spz{0^{n_2}}_{\ancilla}$ for some $n_1$-qubit state $\spz{\phi}$.
\end{definition}

\begin{definition}[History Hamiltonians]\label{defi:history-Hamiltonians}
	Let $U$ be a unitary that acts on $\Hstate = \caH_{\witness} \otimes \caH_{\ancilla}$, with $n_1$ qubits in $\witness$ and $n_2$ qubits in $\ancilla$. Let $S_1,\dotsc,S_{B}$ be a partition of $[n_2]$ (\ie, $[n_2] = \bigsqcup_{i \in [B]} S_i$). Let $U_1,\dotsc,U_{T}$ be a sequence of unitaries such that $U = U_T \dotsc U_1$. 
	
	We define the following Hamiltonians:\\
	
	\noindent$\bullet$ Propagation terms:
	\begin{eqnarray}
	\Hamprop_1 &\equiv& \frac{1}{2} \left( \proj{0}_{\clock(1)} + \proj{10}_{\clock[1,2]} - U_1 \spz{1}\rpz{0}_{\clock(1)} - U_1^\dagger \spz{0}\rpz{1}_{\clock(1)} \right),\\
	\Hamprop_{T} &\equiv& \frac{1}{2} \left( \proj{10}_{\clock[T-1,T]} + \proj{1}_{\clock(T)} - U_{T} \spz{1}\rpz{0}_{\clock(T)} - U_T^\dagger \spz{0}\rpz{1}_{\clock(T)} \right),
	\end{eqnarray}
	and for $t \in \{2,\dotsc,T-1\}$, we set
	\begin{equation}
	\Hamprop_{t}\equiv\frac{1}{2} \left( \proj{10} + \proj{11} - U_{t} \spz{11}\rpz{10} - U_t^\dagger \spz{10}\rpz{11} \right)_{\clock[t-1,t]} \otimes \proj{0}_{\clock(t+1)}.
	\end{equation}
	
	\noindent$\bullet$ For each $t \in [T-1]$, we set the $t$-th stabilizing Hamiltonian as
	\[
	\Hamstab_t = \proj{01}_{\clock[t,t+1]}.
	\]
	\noindent$\bullet$ For each $i \in [B]$, we set the $i$-th initialization Hamiltonian as
	\[
	\Haminit_i \coloneqq (I - \proj{0^{|S_i|}})_{\ancilla(S_i)} \otimes \proj{0}_{\clock(1)}.
	\]
	\noindent$\bullet$ We then define the {\it History} Hamiltonian as
	\[
	H \coloneqq \Hamprop + \Haminit + \Hamstab\;,
	\]
	where
	\[
	\Hamprop \coloneqq \sum_{t \in [T]} \Hamprop_t,\qquad\Haminit \coloneqq \sum_{i \in [B]} \Haminit_i,\qquad \Hamstab \coloneqq \sum_{t \in [T-1]} \Hamstab_t~~.
	\]
	
	We call $H$ the history Hamiltonian of sequence $U_1,\dotsc,U_T$. At times we slightly abuse notation and call $H$ the history Hamiltonian of the unitary $U$ for simplicity.
\end{definition}

We also need the following lemma stating that every state with a low energy with respect to the history Hamiltonian is close to some history state.

\begin{lemma}[{\cite[Theorem~21]{NirkheVY18}}]\label{lemma:close-to-history}
	Let $U$ be a unitary that acts on $\Hstate = \caH_{\witness} \otimes \caH_{\ancilla}$, with $n_1$ qubits in $\witness$ and $n_2$ qubits in $\ancilla$. Let $S_1,\dotsc,S_{B}$ be a partition of $[n_2]$. Let $U_1,\dotsc,U_{T}$ be a sequence of unitaries such that $U = U_T \dotsc U_1$,  and $H$ be the history Hamiltonian of $U$ on $\caH_{\clock} \otimes \caH_{\state}$. Then for every state $\spz{\psi} \in \caH_{\clock} \otimes \caH_{\state}$ such that $\rpz{\psi} H \spz{\psi} \le \delta$, there is a history state $\spz{\eta}$ such that
	\[
	\| \proj{\eta} - \proj{\psi}  \|_1 \le \poly(T,B) \cdot \sqrt{\delta}.
	\]
\end{lemma}


\begin{remark}
	We remark that our the definition of the initialization term $\Haminit$ in our construction above is slightly different from the standard construction used in~\cite[Theorem~21]{NirkheVY18}. But nontheless, their proof easily goes through.
\end{remark}

\subsection{The Quantum Haar Random Oracle Model}

\newcommand{\View}{\mathsf{View}}

We will consider the \emph{Quantum Haar random oracle model} ($\QHROM$), introduced by~\cite{ChenM22}, in which every agent (prover and verifier) gets access to a Haar random oracle $\caG$ acting on $\lambda$ qubits and its inverse $\caG^\dagger$, where $\lambda$ is the so-called \emph{security parameter}.

We denote by $\Haar(N)$ the set of all $N \times N$ unitaries. By $\caG \getsR \Haar(N)$ we mean that $\caG$ is an $N \times N$ unitary drawn from the Haar measure.

\begin{definition}\label{defi:QHROM}
	An interactive protocol $\Pi$ between the prover $\caP$ and verifier $\caV$ is a proof system for a promise problem $L=(\ayes,\ano)$ with completeness $c(n,\lambda)$ and soundness $s(n,t,\lambda)$ in $\QHROM$, if the following holds:
	\begin{description}
		\item[] \textbf{$\caP$ and $\caV$:} $\caP$ and $\caV$ are both given an input $x \in \ayes \cup \ano$. $\caV$ is polynomial-time and outputs a classical bit indicating acceptance or	rejection of $x$, and $\caP$ is unbounded. Both $\caV$ and $\caP$ are given access to a Haar random quantum oracle $\caG$ and its inverse $\caG^\dagger$ that act on $\lambda$ qubits (that is, $\caG \getsR \Haar(2^\lambda)$). Let $n = |x|$. 
		\item[] \textbf{Completeness:}
		{ If $x\in\ayes$, 
			\[
			 \Ex_{\caG \getsR \Haar(2^\lambda)} \Pr[(\caV^{\caG,\caG^\dagger} \leftrightarrows \caP^{\caG,\caG^\dagger} )(x) = 1] \geq c(n,\lambda),
			\]
     where we use $\leftrightarrows$ to denote the interactive nature of the protocol between $\caP$ and $\caV$.
		}

		\item[] \textbf{Soundness:} {If $x\in\ano$, for every $t \in \N$ and any unbounded prover $\caP^*$ making at most $t$ total queries to $\caG$ and $\caG^\dagger$, we have that
			\[
			\Ex_{\caG \getsR \Haar(2^\lambda)} [(\caV^{\caG,\caG^\dagger}  \leftrightarrows (\caP^*)^{\caG,\caG^\dagger} )(x) = 1)] \leq s(n,t,\lambda).
			\]
		}
	
		\item[]  \textbf{Computational zero-knowledge:} {For any $x \in \ayes$ and any polynomial-time $\caV^*$ that receives the inputs $x$ and some state	$\zeta$, there is a polynomial-time quantum channel $\mathcal{S}_{\caV^*}$ that also receives~$x$ and $\zeta$ as inputs such that for all polynomial-time quantum algorithms $D$ that takes a quantum state and outputs a single bit, we have
			$$
			\left|\Pr_{\caG \getsR \Haar(2^\lambda)}\left[D(\View_{\caV^*}((\caV^*)^{\caG,\caG^\dagger} \leftrightarrows \caP^{\caG,\caG^\dagger})(x)) = 1\right] - \Pr_{\caG \getsR \Haar(2^\lambda)}\left[D(\mathcal{S}_{\caV^*}^{\caG,\caG^\dagger}(x,\zeta)) = 1\right] \right| \le \negl(n),
			$$ 
			where $\View_{\caV^*}(\caV^* \leftrightarrows \caP)(x)$ denotes the quantum state of $\caV^*$ at the end of protocol.
		}
	\end{description}
\end{definition}

We remark that in the soundness part, the only restriction on a malicious prover $\caP^*$ is the number of queries it can make to $\caG$ and $\caG^\dagger$. In particular, this means that even if $\caP^*$ has unbounded computational power, as long as it makes a small number of queries to $\caG$ and $\caG^\dagger$, it cannot fool the verifier.

\subsection{Local Proofs and Locally simulable Proofs}

Next, we provide formal definitions of $\LocalQMA$ and $\SimQMA$.

\begin{definition}[$(k,\gamma)$-$\LocalQMA$]\label{defi:localQMA}
	For $k\colon \N \to \N$ and $\gamma \colon \N \to \N$, a promise problem $L = (L_{\yes}, L_{\no})$ is in the complexity class $(k,\gamma)$-$\LocalQMA$ with soundness $s(n)$ and completeness $c(n)$ if there are polynomials $m,p$ such that the following hold:
	
	\begin{description}
		\item[] \textbf{A $k$-local verifier $V_L$:} Let $n = |x|$. There is a verifier $V_L$ that acts as follows:
		\begin{enumerate}
			\item $V_L$ gets access to a $p(n)$-qubit proof $\sigma$ for $L$, it also draws $i \getsR [m(n)]$, $V_L$ then computes in $\poly(n,k,\gamma)$ time a $k$-size subset $S_i \subseteq [p(n)]$ and a $\gamma(n)$-size quantum circuit $C_i$ that is over the Clifford + T gate-set and acts on $k$ qubits. $C_i$ may use $\gamma$ ancilla qubits, with the first ancilla qubit being the output qubit.
			
			\item $V_L$ next applies $C_i$ to the restriction of $\sigma$ on qubits in $S_i$ and measures the first ancilla qubit. $V_L$ accepts if the outcome is $1$ and rejects otherwise.
		\end{enumerate}
		
		\item[] \textbf{Completeness:} If $x \in L_{\yes}$, there is a $p(n)$-qubit state $\sigma$ such that $V_L$ accepts $\sigma$ with probability at least $c(n)$.			
		
		\item[] \textbf{Soundness:} If $x \in L_{\no}$, $V_L$ accepts every $p(n)$-qubit state $\sigma$ with probability at most $s(n)$.
		
		\item[] \textbf{Strongly explicit:} Moreover, we say that $V_L$ is strongly explicit, if $V_L$ computes $S_i$ and $C_i$ in $\poly(\log n,k,\gamma)$ time instead of $\poly(n,k,\gamma)$ time.
	\end{description}
\end{definition}

We will use $(k,\gamma)$-$\LocalQMA_{s,c}$ to denote the class above for notational convenience.

\begin{definition}[$(k,\ell)$-$\SimQMA$]\label{defi:simQMA}
	For $k\colon \N \to \N$ and $\ell \in \N$, a promise problem $L = (L_{\yes}, L_{\no})$ is in the complexity class $(k,\ell)$-$\SimQMA$ with soundness $s(n)$ and completeness $c(n)$ if there are polynomials $m,p$ and a negligible function $\eps \colon \N \to [0,1]$ such that the following hold:
	
	\begin{description}
		\item[] \textbf{An $\ell$-adaptive verifier $V_L$:} Let $n = |x|$. There is a verifier $V_L$ that acts as follows:
		\begin{enumerate}
			\item $V_L$ gets access to a $p(n)$-qubit proof $\sigma$ for $L$, it also draws $\tau_0 \getsR [m(n)]$. Then $V_L$ proceeds in $\ell$ rounds. In the $i$-th round $V_L$ performs some measurements and obtains an outcome $\tau_i$. We use $\tau_{\le i}$ to denote the sequence $\tau_0,\dotsc,\tau_{i}$.
			
			\item At the beginning of the $i$-th round, based on $\tau_{\le i-1}$, $V_L$ computes in $\poly(n,k,\gamma)$ time a subset $S_i \subseteq [p(n)]$ such that $S_i \cap S_j  = \emptyset$ for all $j < i$ and $|S_i| \le k(n)$, together with a POVM $\{ \Pi_{j} \}_{j \in [m(n)]}$ on $|S_i|$ qubits.
			
			\item $V_L$ next measures $\sigma_{S_i}$ with the $\{ \Pi_{j} \}_{j \in [m(n)]}$, and sets $\tau_i = j$ if it sees $\Pi_j$.
			
			\item Finally, $V_L$ decides whether it accepts or not based on the sequence $\tau_{\le \ell}$.
		\end{enumerate}
		
		\item[] \textbf{Simulable completeness:} If $x \in L_{\yes}$, there is a $p(n)$-qubit state $\sigma$ such that:
		\begin{enumerate}
			\item $V_L$ accepts $\sigma$ with probability at least $c(n)$.
			
			\item Let $t \in [\ell]$. For every possible sequence $\tau_{\le t - 1} \in [m(n)]^{t}$, let $S_1,\dotsc,S_t$ be the corresponding query sets of $V_L$ (they are uniquely determined from $\tau_{\le t-1}$) and $S_{\le t} = \bigsqcup_{i \in [t]} S_i$, one can compute the classical description of a density matrix $\sigma'$ in $\poly(n,2^{|S_{\le t}|})$ time, such that $\| \sigma' - \sigma_{S_{\le t}} \|_1 \le \eps(n)$.
		\end{enumerate}
		
		We call the $\sigma$ above the \emph{simulable witness} of $V_L$ given the input $x$.
		
		\item[] \textbf{Soundness:} If $x \in L_{\no}$, $V_L$ accepts every $p(n)$-qubit state $\sigma$ with probability at most $s(n)$.
		
		\item[] \textbf{Strongly explicit:} Moreover, we say that $V_L$ is strongly explicit, if $V_L$ computes $S_i$ and measures $\{ \Pi_j \}_{j \in m(n)}$ in $\poly(\log n,k,\gamma)$ time instead of $\poly(n,k,\gamma)$ time.
		
	\end{description}
\end{definition}

We will also use $(k,\ell)$-$\SimQMA_{s,c}$ to denote the class above for notational convenience. We remark that our definition of $(k,\ell)$-$\SimQMA$ is a generalization of $k$-$\SimQMA$ in~\cite{BroadbentG20}, which corresponds to the non-adaptive case that $\ell = 1$.

\subsection{The Quantum $\PCP$ Conjecture}\label{sec:QPCP}

We first recall the quantum $\PCP$ conjecture~\cite{AharonovALV09,AharonovAV13}.

\begin{conjecture}[$\QPCP$ conjecture]\label{conj:QPCP}
	There are constants $k \in \N$ and $\alpha,\beta$ satisfying $0 < \alpha < \beta \le 1$ such that $(\alpha,\beta)$-$k$-$\LH$ is $\QMA$-complete.
\end{conjecture}

In particular, the following corollary is immediate from the definition of $(\alpha,\beta)$-$k$-$\LH$.

\begin{corollary}\label{cor:QPCP-and-LocalQMA}
	If $\QPCP$ holds, then there are constants $k,\gamma \in \N$ and $c,s \in [0,1]$ satisfying that $s < c$, such that
	\[
	\QMA \subseteq (k,\gamma)\text{-}\LocalQMA_{s,c}.
	\]
\end{corollary}

We will also consider the following slightly stronger version of $\QPCP$.

\begin{conjecture}[$\QPCP_{\negl}$ conjecture]
	There are constants $k \in \N$ and $\beta \in (0,1)$ and a negligible function $\alpha$ such that $(\alpha,\beta)$-$k$-$\LH$ is $\QMA$-complete.
\end{conjecture}

That is, here we require that in the yes case of $k$-$\LH$, the minimum energy is negligible. We find $\QPCP_{\negl}$ plausible since (1) it was proved that $(\alpha,\Theta(n^{-3}))$-$O(1)$-$\LH$ is $\QMA$-complete~\cite{KSV02,KempeR03,KempeKR06} for some $\alpha = n^{-\omega(1)}$, and a (potential) proof for $\QPCP$ via gap amplification is likely to keep $\alpha$ negligible; (2) it is consistent with the situation of classical $\PCP$, where one can even set $\alpha = 0$.

Similarly, we have the following corollary.

\begin{corollary}\label{cor:QPCP-negl-and-LocalQMA}
	If $\QPCP_\negl$ holds, then there are constants $k,\gamma \in \N$ and $s \in (0,1)$, and a function $c \colon \N \to [0,1]$ satisfying that $1-c$ is negligible, such that
	\[
	\QMA \subseteq (k,\gamma)\text{-}\LocalQMA_{s,c}.
	\]
\end{corollary}
	\section{Making Every Local Verifiers Simulable}\label{sec:locl-sim}

\newcommand{\redEncH}{\mathsf{Red}}
\newcommand{\T}{\mathsf{T}}
\newcommand{\QSAT}{\mathsf{QSAT}}
\newcommand{\val}{\mathsf{val}}
\newcommand{\caI}{\mathcal{I}}


In this section, we prove~\autoref{theo:LocalQMA-in-SimQMA}. We first define the following problem $(k,\gamma)$-$\QSAT$.



\begin{definition}[$(k,\gamma)$-$\QSAT$]
	Given integers $n,m \in \N$, $m$ subsets $S_1,\dotsc,S_m \subseteq [n]$, and $m$ $\gamma$-size quantum circuits $C_1,\dotsc,C_{m}$ over the Clifford + T gate-set such that for every $i \in [m]$, $C_i$ acts on $|S_i|$ qubits and uses at most $\gamma$ ancilla qubits. We write an instance $\caI$ as $\caI = (n,m,\{S_i\}_{i \in [m]},\{C_i\}_{i \in [m]})$.
	
	For an $n$-qubit quantum state $\sigma$, we define
	\[
	\val_\caI(\sigma) \coloneqq \Ex_{i \getsR [m]}\Big[ \Pr[C_i(\sigma_{S_i}) = 1] \Big],
	\]
	where $\Pr[C_i(\sigma_{S_i}) = 1]$ denotes the probability that after applying $C_i$ to the restriction of $\sigma$ to $S_i$, measuring the first ancilla qubit, and seeing outcome $1$. We also define
	\[
	\val(\caI) \coloneqq \sup_{\text{$\sigma$ is an $n$-qubit state}} \val_\caI(\sigma).
	\]
\end{definition}

\newcommand{\balpha}{\bar{\alpha}}
\newcommand{\bbeta}{\bar{\beta}}

\subsection{QECCs and Locally simulable QECCs}

For $n > k$, an $[[N,K]]$ quantum error correcting code (QECC) is a mapping
from a  $K$-qubit state~$\spz{\psi}$ into an $N$-qubit
state $\Enc(\spz{\psi})$.  The distance of an $[[N,K]]$ QECC is $D$
if for an arbitrary quantum operation $\mathcal{E}$ acting on
$(D-1)/2$ qubits, the original state $\spz{\psi}$ can be recovered from
$\mathcal{E}(Enc(\spz{\psi}))$, and in this case we call it an $[[N,K,D]]$~QECC. For an $[[N,1,D]]$-QECC and its encoding $\Enc$, we overload notation and for a $k$-qubit system $\phi$ we write  $\Enc(\phi) := \Enc^{\otimes k}(\phi)$.

We begin by recalling the definition of \emph{locally simulable codes} introduced in~\cite{GriloSY19}, which is also used in~\cite{BroadbentG20}. The following definition is from~\cite{BroadbentG20}.

\newcommand{\Sim}{\mathsf{Sim}}
\newcommand{\Pg}{\mathsf{P}}

\begin{definition}\label{D:simulable}
	Let $\caC$ be an $[[N,1,D]]$-QECC  that allows universal quantum computation on the encoded data by applying logical gates from a universal gate set $\mathcal{G}$
	with transversal gates (and possibly with the help of magic states).
    Let $G \in \mathcal{G}$ be a logical gate acting on $k_G$
	qubits, $U^{(G)}_1, \ldots ,U^{(G)}_{\ell}$ be the
	transversal circuit that is applied to the physical qubits of the encoding of a $k_G$-qubit state which  logically applies $G$ to the data through $\ell_G = \poly(N)$ physical gates with the aid of an $m_G$-qubit magic state~$\tau_G$.
	
    We say that $\caC$ is \emph{$s$-simulable} if there exists a deterministic algorithm $\Sim_{\caC}$ whose input is  $G \in \mathcal{G}$,  a value $0 \leq t \leq \ell_G$, and a subset~$S \subseteq [N(m_G + k_G)]$ with $|S| \leq s$. $\Sim_{\caC}(G,t,S)$ then runs in time $\poly(2^N)$, and outputs the classical description of an $|S|$-qubit density matrix $\rho(G,t,S)$ such that for every $k_G$-qubit state $\sigma$
	\begin{equation*}
	\rho(G,t,S) = \left((U^{(G)}_t \cdots U^{(G)}_1) \Enc(\sigma\otimes \tau_G)  (U^{(G)}_{t}
	\cdots U^{(G)}_1)^\dagger\right)_S.
	\end{equation*}
\end{definition}
\begin{remark}\label{R:decomposition-gates}
	In this work,  we take that $\mathcal{G} = \{\CNOT,\Pg,\Had,\T\}$, and use magic state $\spz{\T} \coloneqq \T \spz{+}$ to compute $\T$-gates. Notice that the  gates $U^{(G)}_1,\ldots ,U^{(G)}_{\ell_G}$ are publicly known, and therefore they do not need to be a parameter for $\Sim_{\caC}$.
\end{remark}

The following lemma from~\cite{GriloSY19} says that the folded Steane codes are locally simulable. We refer the readers to~\cite{BroadbentG20} for a simplified proof.

\begin{lemma}[\cite{GriloSY19}]\label{lemma:simulable}
	For every $k > \log(s+3)$, the $k$-fold concatenated Steane code is $s$-simulable.
\end{lemma}

In the rest of the section, we will use $\caC$ to denote the $5$-fold concatenated Steane code and $\Enc$ to be its encoding map. We also set $N = 7^5$ and $D = 3^5$, and note that $\caC$ is an $[[N,1,D]]$-QECC. By~\autoref{lemma:simulable}, it follows that $\caC$ is $28$-simulable.

We note that the encoding algorithm $\Enc$ for $\caC$ can be implemented by the following unitary $U^{\Enc}$: For an $N$-qubit register $\state$. Given $\spz{\psi}_{\state} = \spz{\phi}_{\state(1)} \otimes \spz{0^{N - 1}}_{\state[2,N]}$, we have
\[
U^{\Enc} \spz{\psi}_{\state} = \Enc(\spz{\phi})_{\state}.
\]

The decoding algorithm for $\caC$ can similarly be implemented by $ U^{\Dec} = (U^{\Enc})^{\dagger}$. Furthermore, $U^{\Enc}$ can be written as the product of $\ell_{\Enc} = \poly(N)$ two-local gates as follows
\[
U^{\Enc} = U^{\Enc}_{\ell_{\Enc}} \cdots U^{\Enc}_{1},
\]
similarly, we write
\[
U^{\Dec} = U^{\Dec}_{\ell_{\Dec}} \cdots U^{\Dec}_{1}.
\]

Given a state $\spz{\psi}_{\state}$ on $N$ qubits, let the $\syndrome$ be an $(N-1)$-qubit register. We also define a unitary $U^{\Chk}$ acting on $\Hstate \otimes \Hsyndrome$ such that $U^{\Chk}$ measures the syndromes of the code block in $\state$ to $\syndrome$. We write
\[
U^{\Chk} = U^{\Chk}_{\ell_{\Chk}} \cdots U^{\Chk}_{1},
\]
where all $U^{\Chk}_i$'s are $2$-local.

Given a Clifford gate $G$ acting on $k_G$ qubits, we use $U_1^{(G)},\dotsc,U_{\ell_G}^{(G)}$ to denote the fault-tolerant version of $G$ acting on $k_G \cdot N$ physical qubits, where $\ell_G = \poly(N)$.

\paragraph*{Notation.} For simplicity, for an $n \cdot N$-qubit register $\reg$ and $i \in [n]$, we use $\reg\{i\}$ to denote $\reg[(i-1)\cdot N+1,i \cdot N]$, which is the set of qubits corresponding to the $i$-th block. We also use $\reg\{\ell,r\}$ to denote $\reg[(\ell-1) \cdot N + 1, r \cdot N]$, which is the set of qubits from the $\ell$-th block to the $r$-th block. We say that a subset $S$ has at most $b$ intersections per block with $\reg$, if $|S \cap \reg\{i\}| \le b$ for every $i \in [n]$.

We will need the following additional property of locally simulable codes.

\begin{lemma}\label{lemma:n-copy-sim}
	Let $\caC$ be $s$-simulable and $G \in \caG$. Let $n \in \N$, $\reg$ be an $n \cdot (k_G \cdot N)$-qubit register, and $\reg_{\sf m}$ be an $n \cdot (m_G \cdot N)$-qubit register. For every $j \in [\ell_G]$, let
	\[
	U_j^{(G^{\otimes n})} = \prod_{i \in [n]} (U_j^{(G)})_{\reg\{ (i-1)\cdot k_G+1, i \cdot k_G \}, \reg_{\sf m}\{ (i-1)\cdot m_G+1, i \cdot m_G \}}.
	\]
	
	There is an algorithm $\Sim_{\caC^{\otimes n}}$ that receives as input $G \in \mathcal{G}$, an integer $0 \leq t \leq \ell_G$ and a subset~$S \subseteq (\reg \cup \reg_{\sf m})$ that has $s$ intersections per block with $\reg \cup \reg_{\sf m}$, runs in time $\poly(2^{N},n)$, and outputs the classical description of an $|S|$-qubit density matrix $\rho(G^{\otimes n},t,S)$ such that for every $n\cdot k_G$-qubit state $\sigma$
	\begin{equation*}
	\rho(G^{\otimes n},t,S) = \left((U^{(G^{\otimes n})}_t \cdots U^{(G^{\otimes n})}_1) \Enc(\sigma\otimes \tau_G^{\otimes n \cdot m_G})  (U^{(G^{\otimes n})}_{t}
	\cdots U^{(G^{\otimes n})}_1)^\dagger\right)_S.
	\end{equation*}
	
	In above, $\Enc(\sigma\otimes \tau_G^{\otimes n \cdot m_G})$ puts the encoded $\sigma$ in the register $\reg$ and the encoded $\tau_G^{\otimes n \cdot m_G}$ in the register $\reg_{\sf m}$.
\end{lemma}
\begin{proof}
	We simply set 
	\[
	\rho(G^{\otimes n},t,S) = \otimes_{i \in [n]} \rho(G,t,S \cap (\reg\{i\} \cup \reg_{\sf m}\{i\})),
	\]
	where $\rho(G,t,S \cap (\reg\{i\} \cup \reg_{\sf m}\{i\}))$ is the output of $\Sim_{\caC}$ in \autoref{D:simulable}, and we interpret $S \cap (\reg\{i\} \cup \reg_{\sf m}\{i\})$ as a subset of $[N(m_G + k_G)]$ by shifting.
	
	To see the lemma, we first note that it holds when $\sigma = \otimes_{i \in [n]} \sigma_i$, where each $\sigma_i$ is an $k_G$-qubit state. The lemma follows from the fact that a general $n \cdot k_G$-qubit state $\sigma$ can be written as a linear combination of many product states.
\end{proof}

Let the parameter $t = 0$ and the gate type be $\Had$ in~\autoref{lemma:n-copy-sim} so that $m_\Had = 0$ and $k_\Had = 1$, we have the following corollary.

\begin{corollary}\label{cor:n-copy-sim-no-op}
	Let $\caC$ be $s$-simulable, $n \in \N$, and $\reg$ be an $(n \cdot N)$-qubit register, there is an algorithm $\Sim_{\caC^{\otimes n}}$ that receives a subset $S \subseteq \reg$ that has $s$ intersections per block with $\reg$, and outputs the classical description of a density matrix $\rho(G^{\otimes n},S)$ such that
	\[
	\rho(G^{\otimes n},S) = \Enc(\sigma)_S,
	\]
	for every $n$-qubit state $\sigma$.
\end{corollary}

\subsection{The Encoded Verifier $\Venc$ for $(k,\gamma)$-$\QSAT$}

We consider the following encoded verifier for $(k,\gamma)$-$\QSAT$. In fact, it is the fault-tolerant version.

\begin{construction}{The encoded verifier $\Venc$ for $(k,\gamma)$-$\QSAT$}
	\begin{itemize}
		\item (\textbf{Input.}) $\Venc$ is given a $(k,\gamma)$-$\QSAT$ instance $\caI = (n,m,\{S_i\}_{i \in [m]},\{C_i\}_{i \in [m]})$. Without loss of generality, we can assume that $|S_i| = k$ for every $i \in [m]$. We also require that $m$ is a power of $2$.
		
		\item (\textbf{Witness.}) $\Venc$ expects an $(3 \cdot N \cdot n)$-qubit encoded witness in an $(N \cdot n)$-qubit register $\edata$ and a $(2 \cdot N \cdot n)$-qubit register $\eotp$. Further, let $\nwit = 3 \cdot N \cdot n$, and $\witness = \eotp \cup \edata$.
		
		$\Venc$ expects the $\edata$ register to encode an $n$-qubit witness $\spz{\phi}$ to $\caI$ with a uniformly random quantum  encoded one-time pad in $\eotp$. (The purpose of this one-time pad is to ensure that all encoded qubits are maximally mixed. This is required later when showing $\VHenc$ is simulable.) Formally, the encoding is expected to be
		\begin{equation}\label{eq:def-otp}
		\phi^{\mathsf{otp}} = \frac{1}{2^{2n}} \sum_{a,b \in \bits^n} \Enc( \proj{a,b} \otimes X^{a} Z^{b} \proj{\phi} Z^b X^a)_{\eotp,\edata}.
		\end{equation}
		
		\item (\textbf{Ancilla.}) $\Venc$ also uses $(2\log m + 2\gamma) \cdot N + 3k \cdot (N-1)$ ancilla qubits consisting of $3k \cdot (N-1)$ qubits in register $\echk$, $\log m \cdot N$ qubits in register $\eidx$, $\log m \cdot N$ qubits in register $\emidx$, $\gamma \cdot N$ qubits in $\eT$, and $\gamma \cdot N$ qubits in $\eancC$. For notational convenience, we also use $\echk\{i\}$ to denote $\echk[ (i-1)(N-1) + 1, i (N-1)]$.
		
		Also, all ancilla states are expected to be initialized to $\spz{0}$. We also let $\nanc = (2\log m + 2\gamma) \cdot N + k \cdot (N-1)$, and $\ancilla = \echk \cup \eidx \cup \emidx \cup \eT \cup \eancC$.
		
		In the following, we describe how $\Venc$ operates.
		
		\item (\textbf{Set-up magic states.}) $\Venc$ first applies $\T \cdot \Had$ to the first qubit of $\eT\{i\}$ for every $i \in [\gamma]$. (Note that $\T \cdot \Had \spz{0} = \T \spz{+} = \spz{\T}$.)
		
		\item (\textbf{Encoding.}) $\Venc$ then encodes the registers $\eidx$, $\emidx$, $\eout$, and $\eT$ by $\Enc$. Now, $\idx$ and $\midx$ are simply $\Enc(\spz{0^{\log m}}\spz{0^{\log m}})$.
		
		\item (\textbf{Encoded Hadamard.}) Next, $\Venc$ uses fault tolerance computation (transversal $H$ gate) to apply $H$ to each logical qubits encoded in $\idx$. Now, $\idx$ and $\midx$ are $\Enc(\spz{+}^{\otimes \log m}\spz{0^{\log m}})$.
		
		\item (\textbf{Encoded $\CNOT$.}) $\Venc$ then uses fault tolerance computation (transversal $\CNOT$ gate) to apply $\CNOT$ on corresponding logical qubits in $\idx$ and $\midx$. Now, $\eidx$ and $\emidx$ are $\Enc\left( \frac{1}{\sqrt{m}} \sum_{i=1}^{m} \spz{i}\spz{i} \right)$.
		
		\item (\textbf{Check encoding.}) Next, $\Venc$ applies the following three unitaries $\UcheckE^{(1)}$, $\UcheckE^{(2)}$, and $\UcheckE^{(3)}$, defined as
		\[
		\UcheckE^{(1)} =  \prod_{\tau \in [k]} \sum_{i \in [m]} (U^\Chk)_{\edata\{S_{i,\tau}\},\echk\{i\}}  \otimes \proj{\Enc(i)}_{\eidx},
		\]
		\[
		\UcheckE^{(2)} =  \prod_{\tau \in [k]} \sum_{i \in [m]} (U^\Chk)_{\eotp\{S_{i,\tau}\},\echk\{i\}}  \otimes \proj{\Enc(i)}_{\eidx},
		\]
		\[
		\UcheckE^{(3)} =  \prod_{\tau \in [k]} \sum_{i \in [m]} (U^\Chk)_{\eotp\{S_{i,\tau}+n\},\echk\{i\}}  \otimes \proj{\Enc(i)}_{\eidx},
		\]					
		Namely, this step checks whether the corresponding blocks in $\edata,\eotp$ for qubits in $S_i$ are corrected encoded by measuring all syndromes on the corresponding $3k$ $N$-qubit blocks into $\echk$.
		
		\item (\textbf{Encoded $U_{\sf test}$.}) Now $\Venc$ applies a unitary $\Utestenc$ acting on $\eidx$, $\eotp$, $\edata$ and $\eout$ as follows:
		\[
		\Utestenc = \sum_{i \in [m]} \Enc(V_{C_i})_{\eotp,\edata,\eT,\eancC} \otimes \proj{\Enc(i)}_{\eidx}.
		\]
		
		In above, $\Enc(V_{C_i})$ acts as follows: 
		
		\begin{enumerate}
			\item It first undoes the one-time pad on blocks of $\edata$ corresponding to $S_i$ by transversally applying $\CNOT$ between $\eotp\{u\}$ and $\edata\{u\}$, and then transversally applying $\CZ$ between $\eotp\{u + n\}$ and $\edata\{u\}$, for every $u \in S_i$.
			
			\item Next, it applies the fault-tolerant version of $C_i$ on blocks of $\edata$ corresponding to $S_i$, $\eancC$, and $\eT$. Clifford gates are implemented transversally; $\T$ gates are implemented using the encoded $\spz{\T}$ gates in $\eT$ together with transversal Clifford gates; it uses (encoded) ancilla qubits in $\eancC$.
		\end{enumerate}
		
		\item (\textbf{Decoding $\eancC\{1\}$.}) Then, $\Venc$ decodes $\eancC\{1\}$ by applying $U^{\Dec}$ to $\eancC\{1\}$. 
		
		\item (\textbf{Decision.}) Finally, $\Venc$ measures $\eancC(1)$ and $\echk$ in the computational basis to obtain an outcome $w_1 \in \bits^1$ and $w_2 \in \bits^{k \cdot (N-1)}$, it \textbf{accepts} if $w_1 = 1$ and $w_2 = 0^{k \cdot (N-1)}$, and \textbf{rejects} otherwise.
	\end{itemize}
\end{construction}

$\newline$

We will use $\Venc_\caI$ to denote the above verifier when given the $(k,\gamma)$-$\QSAT$ instance $\caI$. We need the following claim.

\begin{claim}\label{claim:Venc-acc}
	If there is an $\nwit$-qubit state $\sigma$ that makes $\Venc_\caI$ accept with probability $\alpha$, then $\val(\caI) \ge \alpha$.
\end{claim}
\begin{proof}
	We first note that without loss of generality, we can assume $\sigma$ to be a pure state $\spz{\psi}$. Let $\Pi$ be the projection onto the $3n$-dimensional encoded subspace of $\witness$. Let $|\phi\rangle$ be 
	\[
	\spz{\phi} = \frac{ \Pi \spz{\psi}}{\|\Pi \spz{\psi}\|}.
	\]
	
Since $\Venc_\caI$ rejects immediately if the measurement outcome on $\Chk$ is not $0^{k(N-1)}$, we can see that $\Venc_\caI$ accepts $\spz{\phi}$ with probability at least $\alpha$. Now we have
	\[
	\spz{\phi} = \Enc(\spz{\phi_0})
	\]
	for some $3n$-qubit state $\spz{\phi_0}$. Next let $\otp$ be a $2n$-qubit register and $\data$ be an $n$-qubit register. We assume that $\spz{\phi_0}$ is stored in $\otp$ and $\data$. We also assume that each qubit in $\otp$ is encoded as a block in $\eotp$, and each qubit in $\data$ is encoded as a block in $\edata$.
	
	Now we define the following state $\spz{\psi_0}$
	\[
	\spz{\psi_0} = \left(\prod_{i=1}^{n}\CZ_{\otp(i+n),\data(i)}\right) \left(\prod_{i=1}^{n}\CNOT_{\otp(i),\data(i)}\right) \spz{\phi_0}.
	\]	
	From the definition of $\Venc$, it follows that the acceptance probability of $\Venc$ on $\spz{\phi}$ is exactly $\val_\caI(\spz{\psi_0})$. Hence we have $\val(\caI) \ge \val_\caI(\spz{\psi_0}) \ge \alpha$.
\end{proof}

In the following, we first decompose the unitary $\Venc$ into a product of constantly many sub-unitaries and then construct the corresponding history Hamiltonian with respect to~\autoref{defi:history-Hamiltonians}.

\subsubsection{Decomposing $\Venc$ into a Product of Constant Many Sub-Unitaries}\label{sec:decompose-unitary}

\newcommand{\elletest}{\ell_{\sf etest}}

\newcommand{\Magic}{\mathsf{magic}}
\newcommand{\ellmagic}{\ell_{\Magic}}

We let $U^{\Magic}_1 = \Had \otimes I^{\otimes (N-1)}$, $U^{\Magic}_2 = \T \otimes I^{\otimes (N-1)}$, and $\ellmagic = 2$.

We deal with each phase of $\Venc$ separately.
\begin{itemize}
	
	\item[(\textbf{1. Set-up magic states})] For each $j \in [\ellmagic]$, we set
	\[
	\Venc_{1,j} = \prod_{i=1}^{\gamma} (U^{\Magic}_{j})_{\eT\{i\}}.
	\]
	
	\item[(\textbf{2. Encoding})] For each $j \in [\ell_{\Enc}]$, we set
	\[
	\Venc_{2,j} = \prod_{i=1}^{\gamma} (U^{\Enc}_j)_{\eancC\{i\}} \cdot \prod_{i=1}^{\gamma} (U^{\Enc}_j)_{\eT\{i\}} \cdot \prod_{i=1}^{\log m} (U^{\Enc}_j)_{\eidx\{i\}} \cdot \prod_{i=1}^{\log m} (U^{\Enc}_j)_{\emidx\{i\}}.
	\]
	\item[(\textbf{3. Encoded Hadamard})]
	Next, for each $j \in [\ell_H]$, we set
	\[
	\Venc_{3,j} = \prod_{i=1}^{\log m} (U^{(H)}_j)_{\eidx\{i\}}.
	\]
	\item[(\textbf{4. Encoded $\CNOT$})]
	For each $j \in [\ell_{\CNOT}]$, we set
	\[
	\Venc_{4,j} = \prod_{i=1}^{\log m} (U^{(\CNOT)}_j)_{\eidx\{i\}, \emidx\{i\}}.
	\]
	\item[(\textbf{5. Check encoding})]
	For each $j \in [\ell_{\Chk}]$, $i \in [m]$, and $\tau \in [k]$, we set
	\[
	\Venc_{5,(\tau-1) \cdot \ell_{\Chk} + j} = \sum_{i \in [m]} (U^\Chk_j)_{\edata\{S_{i,\tau}\},\echk\{\tau\}} \otimes \proj{\Enc(i)}_{\idx},
	\]
	and
	\[
	\Venc_{5,(\tau-1) \cdot \ell_{\Chk} + j,i} = (U^\Chk_j)_{\edata\{S_{i,\tau}\},\echk\{\tau\}} \otimes \proj{\Enc(i)}_{\idx}.
	\]
	
	Let
	\[
	\Venc_{5,(\tau-1+k) \cdot \ell_{\Chk} + j} = \sum_{i \in [m]} (U^\Chk_j)_{\eotp\{S_{i,\tau}\},\echk\{\tau+k\}} \otimes \proj{\Enc(i)}_{\idx},
	\]
	and
	\[
	\Venc_{5,(\tau-1+2k) \cdot \ell_{\Chk} + j} = \sum_{i \in [m]} (U^\Chk_j)_{\eotp\{S_{i,\tau}+n\},\echk\{\tau+2k\}} \otimes \proj{\Enc(i)}_{\idx}.
	\]
	
	We further set $\Venc_{5,(\tau-1+k) \cdot \ell_{\Chk} + j,i}$ and $\Venc_{5,(\tau-1+2k) \cdot \ell_{\Chk} + j,i}$ similarly, for every $i \in [m]$.
	
	\item[(\textbf{6. Encoded $U_{\sf test}$})] Let $\elletest = c_{\sf test} \cdot \gamma$ for a sufficiently large constant $c_{\sf test } \in \N$. By adding dummy identity gates, we can assume that all $\Enc(V_{\Pi_i})$ are implemented by $\elletest$ gates, and we write
	\[
	\Enc(V_{\Pi_i}) = U^{(V_{\Pi_i})}_{\elletest} \cdots U^{(V_{\Pi_i})}_{1}.
	\]
	
	Finally, for each $j \in [\elletest]$ and $i \in [m]$, we set
	\[
	\Venc_{6,j} = \sum_{i \in [m]} U^{(V_{\Pi_i})}_{j} \otimes \proj{\Enc(i)}_{\eidx},
	\]
	and
	\[
	\Venc_{6,j,i} = U^{(V_{\Pi_i})}_{j} \otimes \proj{\Enc(i)}_{\eidx}.
	\]
	
	\item[(\textbf{7. Decoding $\eancC\{1\}$})]
	For each $j \in [\ell_{\Dec}]$, we set
	\[
	\Venc_{7,j} = (U^{\Dec}_{j})_{\eancC\{1\}}.
	\]
\end{itemize}

Now, for $i \in [7]$, we let $\ell^{(i)}$ be the $i$-th number in the sequence $\ellmagic,\ell_\Enc,\ell_H,\ell_\CNOT,3k \cdot \ell_\Chk,\elletest,\ell_\Dec$. 

Note that the above decomposed $\Venc$ into a product of $T = \sum_{i \in [7]} \ell^{(i)}$ sub-unitaries as follows:
\begin{align}
\Venc &= ( \Venc_{7,\ell^{(7)}} \cdots \Venc_{7,1} ) \cdot ( \Venc_{6,\ell^{(6)}} \cdots \Venc_{6,1} ) \cdot ( \Venc_{5,\ell^{(5)}} \cdots \Venc_{5,1} ) \label{eq:Venc-seq} \\
&\phantom{{} = {}} ( \Venc_{4,\ell^{(4)}} \cdots \Venc_{4,1} ) \cdot ( \Venc_{3,\ell^{(3)}} \cdots \Venc_{3,1} ) \cdot ( \Venc_{2,\ell^{(2)}} \cdots \Venc_{2,1} ) \cdot (\Venc_{1,\ell^{(1)}} \cdots \Venc_{1,1}).\notag
\end{align}

\subsubsection{Constructing the History Hamiltonian}

\newcommand{\HamEnc}{H^{\mathsf{enc}}}

We decompose the $(2\log m + 2\gamma) \cdot N + k \cdot (N-1)$ ancilla qubits into $5$ subsets:
\begin{align*}
S_1 &= \{ \eidx(i) \}_{i \in [\log m \cdot N]},~~S_2 = \{ \emidx(i) \}_{i \in [\log m \cdot N]},\\
S_3 &= \{ \eancC(i) \}_{i \in [\gamma \cdot N]},~~S_4 = \{ \eT(i) \}_{i \in [\gamma \cdot N]},~~S_5 = \{ \echk(i) \}_{i \in [k \cdot (N-1)]}.
\end{align*}

Let $\Hamhis$ be the history Hamiltonian of $\Venc$ with $B = 5$ and $S_1,\dotsc,S_5$ as above. We will use $\state$ to denote $\witness \cup \ancilla$, and $\clock$ to denote the clock registers for $\HamEnc$. Let
\[
\Hamout = (I - \proj{10^{k(N-1)}})_{\eancC(1),\echk} \otimes \proj{1}_{\clock(T)}
\]

Finally,  let
\begin{align*}
\HamEnc  = \Hamhis + \Hamout = \Hamprop + \Haminit + \Hamstab + \Hamout.
\end{align*}
We use $\Hamprop_{i,j}$ to denote the term in $\Hamprop$ that corresponds to the unitary $\Venc_{i,j}$. Then we have
\[
\Hamprop = \sum_{i \in [7]} \sum_{j \in [\ell^{(i)}]} \Hamprop_{i,j},
\]
and
\[
\HamEnc = \sum_{i \in [7]} \sum_{j \in [\ell^{(i)}]} \Hamprop_{i,j} + \sum_{i \in [B]} \Haminit_i + \sum_{t \in [T-1]} \Hamstab_t + \Hamout + \Hamcheck.
\]

We call $\HamEnc$ the \emph{encoded Hamiltonian} of $\caI$. We will use $\HamEnc(\caI)$ to denote the encoded Hamiltonian of $\caI$.

\newcommand{\aM}{a^{\sf M}}
\newcommand{\bM}{b^{\sf M}}
\newcommand{\cM}{c^{\sf M}}

Let $M = \sum_{i \in [7]} \ell^{(i)} + B + T - 1 + 2 = 2T + B + 1$ be the number of terms in $\HamEnc$. We note that $\ell^{(1)},\ell^{(2)},\ell^{(3)},\ell^{(4)},\ell^{(7)}$ are constants, $\ell^{(5)} = 3k \cdot \ell_{\Chk}$, and $\ell^{(6)} = c_{\sf test} \cdot \gamma$. Hence, there are three constants $a^M,b^M,c^M \in \N$ such that 
\[
M = \aM \cdot k + \bM \cdot \gamma + \cM.
\]

Now, notice that except for the $\Hamprop_{5,j}$'s and the $\Hamprop_{6,j}$'s, all other $M - \ell^{(5)} - \ell^{(6)}$ terms above are at most $O(\log m)$-local.

\subsubsection{Analysis of The Extended Hamiltonian $\HamEnc$}

Here we prove the following lemma relating the minimum energy of $\HamEnc(\caI)$ to $\caI$.

\begin{lemma}\label{lemma:HamEnc-analysis}
	There is a universal polynomial $p$ such that for every $(k,\gamma)$-$\QSAT$ instance $\caI$, 
	\[
	\left( \frac{1 - \val(\caI)}{p(k,\gamma)} \right)^2 \le \lmin(\HamEnc(\caI)) \le 1 - \val(\caI).
	\]	  
\end{lemma}

\autoref{lemma:HamEnc-analysis} is an immediate corollary of the following~\autoref{lemma:enc-soundness} and~\autoref{lemma:enc-completeness}. 

\begin{lemma}[Soundness]\label{lemma:enc-soundness}
	There is a universal polynomial $p$ such that for every $(k,\gamma)$-$\QSAT$ instance $\caI$, if there is a state $\spz{\Phi} \in \Hclock \otimes \Hstate$ such that $\rpz{\Phi} \HamEnc(\caI) \spz{\Phi} \le \delta$, then $\val(\caI) \ge 1 -  p(k,\gamma) \cdot \sqrt{\delta}$.
\end{lemma}
\begin{proof}
	By \autoref{lemma:close-to-history}, there is a universal polynomial $q$ and a history state $\spz{\Psi}$ such that
	\begin{equation}\label{eq:close-history}
	\| \proj{\Psi} - \proj{\Phi} \|_{1} \le q(B,T) \cdot \sqrt{\delta}.
	\end{equation}
	
	Recall that $B = 5$ and $T = O(k + \gamma)$. So there is a universal polynomial $p_1$ such that $q(B,T) \le p_1(k,\gamma)$.
	
	Let $\HamEnc = \HamEnc(\caI)$ for simplicity. From~\eqref{eq:close-history} and the fact that $\proj{\Psi}$ is a history state, it follows that
	\[
	\rpz{\Psi} \HamEnc \spz{\Psi} = \rpz{\Psi} \Hamout \spz{\Psi} \le p_1(k,\gamma) \cdot \sqrt{\delta}.
	\]
	
	We write
	\[
	\spz{\Psi} = \frac{1}{\sqrt{T+1}} \sum_{t=0}^{T} \spz{\unary(t)}_{\clock} \otimes \spz{\psi_t}_{\state},
	\]
	where $\spz{\psi_t} = \Venc_{[1\dotsc t]} \spz{\psi_{0}}$, and $\spz{\psi_0} = \spz{\phi}_{\witness} \otimes \spz{0^{\nanc}}_{\ancilla}$.
	
	Recall that
	\[
	\Hamout = (I - \proj{10^{k(N-1)}})_{\eancC(1),\echk} \otimes \proj{1}_{\clock(T)};
	\]
	we have
	\[
	\rpz{\Psi} \Hamout \spz{\Psi} = \frac{1}{T} \cdot \rpz{\psi_T} (I - \proj{10^{k(N-1)}})_{\eancC(1),\echk} \spz{\psi_T}.
	\]
	
	From the definition of $\Venc$, we can see $\rpz{\psi_T} (I - \proj{10^{k(N-1)}})_{\eancC(1),\echk} \spz{\psi_T}$ is the probability that $\Venc$ rejects the witness $\spz{\phi}$.
	
	Hence, it follows that $\Venc$ rejects $\spz{\phi}$ with probability at most $T \cdot p_1(k,\gamma)$. Recall that $T = O(k + \gamma)$, so there is a universal polynomial $p$ such that $T \cdot p_1(k, \gamma) \le p(k, \gamma)$. Then by~\autoref{claim:Venc-acc}, it follows that $\val(\caI) \ge 1 - p(k, \gamma) \cdot \sqrt{\delta}$.
\end{proof}

\begin{lemma}[Completeness]\label{lemma:enc-completeness}
	For every $(k,\gamma)$-$\QSAT$ instance $\caI$, it holds that
	\[
	\lmin(\HamEnc(\caI)) \le 1 - \val(\caI).
	\]
\end{lemma}
\begin{proof}
	Let $\spz{\psi}$ be the quantum state such that $\val_\caI(\spz{\psi}) = \val(\caI)$. We let $\psi^{\otp}$ be as in $\Venc$, and let $\spz{\Psi}$ be the history state such that 
	\[
	\spz{\Psi} = \frac{1}{\sqrt{T+1}} \sum_{t=0}^{T} \spz{\unary(t)}_{\clock} \otimes \spz{\psi_t}_{\state},
	\]
	where $\spz{\psi_t} = \Venc_{[1\dotsc t]} \spz{\psi_{0}}$, and $\spz{\psi_0} = \spz{\phi}_{\witness} \otimes \spz{0^{\nanc}}_{\ancilla}$.
	
	Since it is a history state, we have
	\[
	\rpz{\Psi} \HamEnc \spz{\Psi} = \rpz{\Psi} \Hamout \spz{\Psi}.
	\]
	
	By a similar calculation as in the proof of~\autoref{lemma:enc-completeness}, we further have
	\[
	\rpz{\Psi} \Hamout \spz{\Psi} = \frac{1}{T} \cdot \rpz{\psi_T} (I - \proj{10^{k(N-1)}})_{\eancC(1),\echk} \spz{\psi_T}.
	\]
	
	Since $\rpz{\psi_T} (I - \proj{10^{k(N-1)}})_{\eancC(1),\echk} \spz{\psi_T}$ is the probability that $\Venc$ rejects the witness $\spz{\phi}$, we have
	\[
	\rpz{\Psi} \HamEnc \spz{\Psi} \le \frac{1}{T} \cdot (1 - \val(\caI)).
	\]
	
	The lemma follows from the fact that $T \ge 1$.
\end{proof}


\subsection{The Encoded History Verifier $\VHenc$ for $(k,\gamma)$-$\History$-$\EncprojLH$}

We now define the problem $(k,\gamma)$-$\History$-$\EncprojLH$.

\begin{definition}[$(k,\gamma)$-$\History$-$\EncprojLH_{\alpha,\beta}$] For $\alpha,\beta \colon \N \to [0,1]$, the $(k,\gamma)$-$\History$-$\EncprojLH_{\alpha,\beta}$ problem is defined as follows: Given a Hamiltonian $\HamEnc = \HamEnc(\caI)$ for some $(k,\gamma)$-$\QSAT$ instance $\caI = (n,n,\{S_i\}_{i \in [n]},\{C_i\}_{i \in [n]})$, where $n$ is a power of $2$.\footnote{That is, here we assume that the $(k,\gamma)$-$\QSAT$ instance has the same number of qubits and local checks.} The goal is to distinguish between the following two cases:
	\begin{itemize}
		\item[(\textbf{Yes})] $\lmin(\HamEnc) \le \alpha(n)$ or 
		\item[(\textbf{No})] $\lmin(\HamEnc) \ge \beta(n)$.
	\end{itemize}
\end{definition}

Next, we introduce the encoded history verifier $\VHenc$ for $(k,\gamma)$-$\History$-$\EncprojLH$.

\begin{construction}{The encoded history verifier $\VHenc$ for $(k,\gamma)$-$\History$-$\EncprojLH$}
	\begin{itemize}
		\item (\textbf{Input.}) $\VHenc$ is given a $(k,\gamma)$-$\QSAT$ instance $\caI = (n,n,\{S_i\}_{i \in [n]},\{C_i\}_{i \in [n]})$, where $n$ is a power of $2$. Without loss of generality, we can assume that $|S_i| = k$ for every $i \in [n]$. Let $\HamEnc = \HamEnc(\caI)$.
		
		\item (\textbf{Witness.}) $\VHenc$ is given a witness $\spz{\Phi} \in \Hclock \otimes \Hstate$ that is supposed to be a history state of $\Venc$.
		
		$\VHenc$ operates as follows:
		
		\item It picks a term $J$ in $\HamEnc$  uniformly at random, and proceeds according to the following two cases:
		
		\begin{enumerate}
			\item If $J$ is not one of the $\Hamprop_{5,j}$'s or the $\Hamprop_{6,j}$'s, it measures $\{J,I-J\}$ and \textbf{rejects} if it sees $J$, and \textbf{accepts} otherwise.
			
			\item Otherwise, $J = \Hamprop_{t,j}$ for some $t \in \{5,6\}$ and $j \in [\ell^{(t)}]$. We write
			\[
			J = \sum_{i \in [n]} J_i \otimes \proj{\Enc(i)}_{\eidx},
			\]
			where $J_i = \Hamprop_{t,j,i}$, for each $i \in [n]$. 
			
			$\VHenc$ first measures $\eidx$ in the basis $\{ \Enc(i) \}_{i \in [n]}$ to obtain $i \in [n]$, and then measures $\{ J_i,I-J_i \}$. It \textbf{rejects} if it sees $J_i$ and \textbf{accepts} otherwise. (We can add additional states orthogonal to $\{ \Enc(i) \}_{i \in [n]}$ to make it a complete basis over qubits in $\eidx$. If those added orthogonal vectors are obtained as measurement outcomes, then $\VHenc$ \textbf{accepts}.)
		\end{enumerate}
	\end{itemize}
\end{construction}

\subsubsection{Analysis of $\VHenc$}

The following lemma shows that the rejecting probability of $\VHenc$ is proportional to the energy of the given state regarding $\HamEnc(\caI)$.

\begin{lemma}\label{lemma:enc-local}
	$\VHenc$ rejects $\spz{\Phi} \in \Hclock \otimes \Hstate$ with probability
	\[
	\frac{1}{M} \rpz{\Phi} \HamEnc \spz{\Phi}.
	\]
\end{lemma}
\begin{proof}
	From the definition of $\VHenc$, when $J$ is not one of the $\Hamprop_{5,j}$'s or the $\Hamprop_{6,j}$'s, $\VHenc$ rejects with probability $\rpz{\Phi} J \spz{\Phi}$.
	
	Otherwise, we can write
	\[
	J = \sum_{i \in [n]} J_i \otimes \proj{\Enc(i)}_{\eidx},
	\]
	
	Let $\other = (\clock \cup \state) \setminus \eidx .$ 
	We write
	\[
	\spz{\Phi} = \sum_{i=1}^{m} \sqrt{p_i} \spz{\eta_i}_{\other} \spz{\Enc(i)}_{\eidx} + \sqrt{q} \spz{\Psi}_{\other,\eidx}.
	\]
	
	One can see in this case, $\VHenc$ rejects with probability
	\[
	\sum_{i \in [n]} p_i \cdot \rpz{\eta_i} J_i \spz{\eta_i} = \rpz{\Phi} J \spz{\Phi}.
	\]
	
	Hence, overall $\VHenc$ rejects with probability
	\[
	\frac{1}{M} \rpz{\Phi} \HamEnc \spz{\Phi}. \qedhere
	\]
\end{proof}

\subsubsection{$\VHenc$ is Simulable}

\begin{theorem}\label{theo:sim}
	For every $k,\gamma \in \N$ and every $\alpha,\beta \colon \N \to [0,1]$ such that $\alpha$ is negligible, it holds that:
	\[
	\text{$(k,\gamma)$-$\History$-$\EncprojLH_{\alpha,\beta}$ $\subseteq$ $(O(\log n),2)$-$\SimQMA_{1-\beta(n)/M,1-\alpha(n)/M}$},
	\]
	where $M = \aM \cdot k + \bM \cdot \gamma + \cM$.
\end{theorem}

We say that two states $\phi$ and $\psi$ are perfectly distinguishable if $F(\phi,\psi) = 0$.\footnote{Equivalently, there is a quantum algorithm that always outputs $0$ (resp. $1$) when given $\phi$ (resp. $\psi$).} We will need the following standard fact.

\begin{lemma}\label{lemma:cond-perfect-dis}
	For two perfectly distinguishable states $\sigma$ and $\rho$, let $\sigma = \sum_{i \in [n]} \alpha_i \proj{\phi_i}$ and $\rho = \sum_{j \in [m]} \beta_j \proj{\psi_j}$ be any decomposition of $\sigma$ and $\rho$ such that $\alpha_i,\beta_j \in (0,1]$, $\{\phi_i\}$ and $\{\psi_j\}$ are both orthonormal, we have $\innerprod{\phi_i}{\psi_j} = 0$ for every $i,j \in [n] \times [m]$.
\end{lemma}
\begin{proof}
	Recall that $F(\rho,\sigma) = \left(\tr \sqrt{\sqrt{\rho} \sigma \sqrt{\rho}}\right)^2$. Since $\sigma$ and $\rho$ are perfectly distinguishable, we have $\tr \sqrt{\sqrt{\rho} \sigma \sqrt{\rho}} = 0$. Note that (1) for a positive matrix $A\ge 0$, $\tr A = 0 \Leftrightarrow A = 0$ and (2) $\sqrt{\rho} \sigma \sqrt{\rho}\ge 0$ since both $\sigma,\rho$ are positive, we further have $ \sqrt{\sqrt{\rho} \sigma \sqrt{\rho}} = 0$, and consequently $\sqrt{\rho} \sigma \sqrt{\rho} = 0$ and hence $\tr(\sqrt{\rho} \sigma \sqrt{\rho}) = 0$.
	
	Using the decomposition of $\sigma$ and $\rho$, using $\tr(\sqrt{\rho} \sigma \sqrt{\rho})=0$, we have
	\begin{align*}
	\tr(\sqrt{\rho} \sigma \sqrt{\rho}) &= \tr\left[ \sum_{i \in [m]} \sqrt{\beta_i} \proj{\psi_i} \cdot \sum_{j \in [n]} \alpha_j \proj{\phi_j} \cdot \sum_{k \in [m]} \sqrt{\beta_k} \proj{\psi_k} \right]\\
	&= \sum_{i,k \in [m]} \tr\left[\spz{\psi_i} \rpz{\psi_k}\right] \cdot \sum_{j \in [n]} \left( \sqrt{\beta_i \beta_k} \cdot \alpha_j \cdot \innerprod{\psi_i}{\phi_j} \cdot \innerprod{\phi_j}{\psi_k} \right).\\
	&= \sum_{i \in [m]} \sum_{j \in [n]} \left( \beta_i \cdot \alpha_j \cdot \innerprod{\psi_i}{\phi_j} \cdot \innerprod{\phi_j}{\psi_i} \right).\\
	&= \sum_{i \in [m]} \sum_{j \in [n]} \left( \beta_i \cdot \alpha_j \cdot \|\innerprod{\psi_i}{\phi_j}\|^2\right)\\
	&= 0
	\end{align*}
	
	Since all the $\alpha_i$ and $\beta_j$ are positive, it follows that $\innerprod{\phi_i}{\psi_j} = 0$ for every $i,j \in [n] \times [m]$, which completes the proof.
	
\end{proof}

We next prove the following lemma. 

\begin{lemma}\label{lemma:ECC-fact}
	For two $n$-qubit pure states $\spz{\phi}$ and $\spz{\psi}$, and a subset $S \subseteq [n]$, if $\phi_{\bar{S}}$ and $\psi_{\bar{S}}$ are perfectly distinguishable, then $\left(\spz{\phi} \rpz{\psi}\right)_{S} = 0$.
\end{lemma}
\begin{proof}
	We write 
	\[
	\spz{\phi} = \sum_{i} \alpha_i \cdot \spz{\phi_i^{\sf A}} \otimes \spz{\phi_i^{\sf B}},
	\]
	where $\spz{\phi_i^{\sf A}}$ and $\spz{\phi_i^{\sf B}}$ are pure quantum states on $S$ and $\bar{S}$, respectively, and $\alpha_i \in (0,1]$. Moreover, $\{ \spz{\phi_i^{\sf A}} \}$ and $\{ \spz{\phi_i^{\sf B}} \}$ are both orthonormal.
	
	Similarly, we write
	\[
	\spz{\psi} = \sum_{j} \beta_j \cdot \spz{\psi_j^{\sf A}} \otimes \spz{\psi_j^{\sf B}},
	\]
	where $\spz{\psi_i^{\sf A}}$ and $\spz{\psi_i^{\sf B}}$ are pure quantum states on $S$ and $\bar{S}$, respectively, and $\beta_i \in (0,1]$. Moreover, $\{ \spz{\psi_i^{\sf A}} \}$ and $\{ \spz{\psi_i^{\sf B}} \}$ are both orthonormal.
	
	Note that $\phi_{\bar{S}} = \sum_{i} \alpha_i^2 \cdot \phi_i^{\sf B}$ and $\psi_{\bar{S}} = \sum_{j} \beta_j^2 \cdot \psi_j^{\sf B}$. Since they are perfectly distinguishable, apply~\autoref{lemma:cond-perfect-dis}, it follows that for every $i,j$, $\innerprod{\phi_i^{\sf B}}{\psi_j^{\sf B}} = 0$.
	
	Hence,
	\[
	\left(\spz{\phi} \rpz{\psi}\right)_{S} = \sum_{i,j} \alpha_i\beta_j \cdot \spz{\phi_i^{\sf A}}\rpz{\psi_j^{\sf A}} \cdot \innerprod{\phi_i^{\sf B}}{\psi_j^{\sf B}} = 0. \qedhere
	\]
\end{proof}

\begin{corollary}\label{cor:Enc-ab-vanish}
	Let $\reg$ be an $n \cdot N$-qubit register, and $S \subseteq \reg$ be such that $S$ has at most $(D-1)/2$ intersections per block with $\reg$. For every two distinct string $a,b \in \bits^n$, we have
	\[
	\Enc(\spz{a}\rpz{b})_S = 0.
	\]
\end{corollary}
\begin{proof}
	By the definition of $\Enc$, we have
	\[
	\Enc(\spz{a}\rpz{b})_S = \bigotimes_{i \in [n]} \Enc(\spz{a_i} \rpz{b_i})_{S \cap \reg\{i\}}
	\]
	
	
	Since $a \ne b$, there exists $i \in [n]$ such that $a_i = b_i$. We now fix such an $i$. We note that $\Enc(\spz{a_i})_{\bar{S} \cap \reg\{i\}}$ and $\Enc(\spz{b})_{\bar{S} \cap \reg\{i\}}$ are perfectly distinguishable (since one can decode $a_i$ or $b_i$ from them). Hence, it follows that $\Enc(\spz{a_i}\rpz{b_i})_{S \cap \reg\{i\}} = 0$ by~\autoref{lemma:ECC-fact}, which further implies $\Enc(\spz{a}\rpz{b})_S = 0 $.
\end{proof}

\begin{lemma}
	Let $k,\gamma \in \N$ be two constants, $n$ be a power of $2$, $\caI = (n,n,\{S_i\}_{i \in [n]},\{C_i\}_{i \in [n]})$ be a $(k,\gamma)$-$\QSAT$ instance, $\HamEnc = \HamEnc(\caI)$, and $\alpha = 1 - \val(\caI)$. Let $\spz{\phi}$ be an $n$-qubit pure state such that $\val_\caI(\phi) = \val(\caI)$, and let $\spz{\Phi} \in \Hclock \otimes \Hstate$ be the history state for $\Venc$ given the $3 \cdot n\cdot N$-qubit witness $\phi^{\otp}$ that is defined according to~\eqref{eq:def-otp}. For every $t \in \zeroTon{T}$, and for any subset $S \subseteq \clock \cup \state$ such that the following holds
	\begin{itemize}		
		\item $S$ has at most $10$ intersections per block with $\edata,\eotp, \emidx,\eancC$, and $\eT$.
	\end{itemize}
	
	Let $\spz{\phi_t} = \Venc_{[1,t]} \Enc(\spz{\phi})_{\edata} \otimes \spz{0^{|\ancilla|}}_{\ancilla}$, one can compute an $|S|$-qubit state $\rho$ such that
	\[
	\left\|\left(\phi_t \right)_{S} - \rho\right\|_1 \le \alpha
	\]
	in $\poly(n,2^{|S|})$ time.
\end{lemma}
\begin{proof}
	
	For every $u \in [7]$ and $j \in [\ell^{(u)}]$, we define $t_{u,j} = j + \sum_{i=1}^{u-1} \ell^{(i)}$. That is, $t_{u,j}$ is the index of $\Venc_{u,j}$ in the sequence of unitaries~\eqref{eq:Venc-seq} that computes $\Venc$.
	
	We write
	\[
	\spz{\Psi} = \frac{1}{\sqrt{T+1}} \sum_{t=0}^{T} \spz{\unary(t)}_{\clock} \otimes \spz{\phi_t}_{\state}\quad,
	\]
	where $\spz{\phi_t} = U_{[1,t]} \spz{\phi_{0}}$, and $\spz{\phi_0} = \spz{\phi^{\otp}} \otimes \spz{0^{|\ancilla|}}_{\ancilla}$. Recall that
	\[
	\spz{\phi^{\otp}} = \frac{1}{2^{n}} \sum_{a,b \in \bits^n} \Enc\left( \spz{a,b} \otimes X^{a} Z^{b} \spz{\phi}\right)_{\eotp,\edata}.
	\]
	
	
	
	We now prove the lemma by considering the following cases separately. We remark that in the first three cases, we can indeed compute $\left( \phi_t \right)_{S}$ exactly. Only for the last case, we can only compute an approximation to $\left( \phi_t \right)_{S}$.
	
	\paragraph*{Case 1: $t \in [0, t_{4,\ell^{(4)}}]$.} During this period, no gates have been applied to $\witness$ yet, hence qubits in $\witness$ and qubits in $\ancilla$ are unentangled. So we have
	\[
	\left(\phi_t \right)_{S \cap (\witness \cup \ancilla)} = \left(\phi_t \right)_{S \cap \witness} \otimes \left(\phi_t \right)_{S \cap \ancilla}.
	\]
	
	$\left(\phi_t \right)_{S \cap \witness}$ can be computed exactly in $\poly(2^{|S|})$ time by noting that qubits in $\witness$ are correctly encoded and then applying~\autoref{cor:n-copy-sim-no-op}. Also, $\left(\phi_t \right)_{S \cap \ancilla}$ can be computed by a straightforward classical simulation in $\poly(n)$ time since $|\ancilla| \le O(\log n)$. Hence, both computations can be done in $\poly(n,2^{|S|})$ time.

	\paragraph*{Notation and analysis for Case 2, 3, 4.} From now on, we can assume $t \ge t_{4,\ell^{(4)}}$. Let us set up the notation and perform preliminary analysis that will be useful for the remaining three cases. Let $\other = \witness \setminus (\idx \cup \midx)$ for notational convenience, we write
	\[ 
	\spz{\phi_t} = \frac{1}{\sqrt{n}} \sum_{i \in [n]} \Enc(\spz{i})_{\eidx} \Enc(\spz{i})_{\emidx} \spz{\phi_{t,i}}_{\other}.
	\]
	
	We also have
	\[
	\phi_t = \frac{1}{n} \sum_{i,j \in [n]} \Enc(\spz{i}\rpz{j}) \otimes \Enc(\spz{i}\rpz{j}) \otimes \spz{\phi_{t,i}} \rpz{\phi_{t,j}}.
	\]

	From~\autoref{cor:Enc-ab-vanish} and the fact that $S$ has at most $10$ intersections per block with $\emidx$ and $10 \le (D-1)/2 = (3^5 - 1) / 2$, it follows that
	\[
	(\phi_t)_S = \frac{1}{n} \sum_{i \in [n]} \Enc(\proj{i})_{S \cap \eidx} \otimes \Enc(\proj{i})_{S \cap \emidx} \otimes (\proj{\phi_{t,i}})_{S \cap \other}.
	\]
	
	Let $S' = S \cap \other$. Note that $|\eidx| = |\emidx| \le O(\log n)$, for every $i \in [n]$, we can compute $\Enc(\proj{i})_{S \cap \eidx} \otimes \Enc(\proj{i})_{S \cap \emidx}$ in $\poly(n)$-time by a classical simulation. Hence, it suffices to compute $(\phi_{t,i})_{S'}$ for every $i \in [n]$. 
	
	\paragraph*{Case 2: $t \in (t_{4,\ell^{(4)}}, t_{5,\ell^{(5)}}]$.} During this period, $\Venc$ checks whether the encoding of the relevant qubits in $\witness$ are correct. Let $j \in [\ell^{(5)}]$ be such that $t = t_{5,j}$. We also let $\tau \in [3k]$ such that $j \in [(\tau-1)\ell_{\Chk}+1, \tau \ell_{\Chk}]$ (Recall that $\ell^{(5)} = 3k \cdot \ell_{\Chk}$.) Let $\mathcal{L}$ be the list 
	\[
	(\edata\{S_{i,1}\},\dotsc,\edata\{S_{i,k}\},\eotp\{S_{i,1}\},\dotsc,\eotp\{S_{i,k}\},\eotp\{S_{i,1}+n\},\dotsc,\eotp\{S_{i,k}+n\}).
	\]
	
	
	
	That is, for each $i \in [n]$, $\spz{\phi_{t,i}}$ records the state that $\Venc$ have checked the encodings of $\mathcal{L}_1,\dotsc,\mathcal{L}_{\tau - 1}$, and is currently checking the encoding of $\mathcal{L}_\tau$. Since all qubits in $\witness$ are correctly encoded, we have 
	\[
	(\phi_{t,i})_{\echk\{u\}} = \proj{0^{N-1}}
	\]
	for every $u \in ([3k] \setminus \{\tau\})$. 
	
	Let $t^0 = t_{4,\ell^{(4)}}$, the only difference between $\phi_{t,i}$ and $\phi_{t^0,i}$ is the unfinished checking and its partial results on $\edata\{S_{i,\tau}\}$ and $\echk[(\tau-1)(N-1)+1,\tau (N-1)]$. Since we already know how to compute $(\phi_{t^0,i})_{S'}$, now we only need to compute
	\[
	(\phi_{t,i})_{S' \cap \left(\edata\{S_{i,\tau}\} \cup \echk[(\tau-1)(N-1)+1,\tau (N-1)]\right)}.
	\]
	Note that in $\phi^{\otp}$, every block is an encoding of the maximally mixed single-qubit state. Thus, the above can be directly computed by a simulation in $\poly(n)$ time.
	
	\paragraph*{Case 3: $t \in (t_{5,\ell^{(5)}}, t_{6,\ell^{(6)}}]$.} By previous discussions, it suffices to compute $(\phi_{t,i})_{S'}$ for every $i \in [n]$.
	
	Note that since all the qubits in $\witness$ are correctly encoded and the checking phase is already completed, $(\phi_{t,i})_{S' \cap \echk} = \proj{0^{|S' \cap \echk|}}$. So we only need to compute $(\phi_{t,i})_{S' \cap (\other \setminus \echk)}$ for every $i \in [n]$.
	
	By our assumption on $S$, we know that $S'$ has at most $10$ intersections per block with $\other \setminus \echk$. Since $\caC$ is $28$-simulable, we can compute $(\phi_{t,i})_{S' \cap (\other \setminus \echk)}$ in $\poly(n, 2^{|S|})$ time by applying~\autoref{lemma:n-copy-sim}.
	
	\paragraph*{Case 4: $t \in (t_{6,\ell^{(6)}}, t_{7,\ell^{(7)}}]$.} Let $t_0 = t_{6,\ell^{(6)}}$. The only difference between $\phi_{t}$ and $\phi_{t_0}$ is the partial decoding on $\eancC\{1\}$. Hence we have
	\[
	(\phi_{t})_S = (\phi_{t_0})_{S \cap (\state \setminus \eancC\{1\})} \otimes (\phi_{t})_{S \cap \eancC\{1\}}.
	\]
	
	Since $t_0$ belongs to case $3$ and $S \cap (\state \setminus \eancC\{1\})$ also satisfies our requirement since it has less elements than $S$, we can compute $(\phi_{t_0})_{S \cap (\state \setminus \eancC\{1\})}$ in $\poly(n,2^{|S|})$ time. So we only need to approximate $(\phi_{t})_{S \cap \eancC\{1\}}$. 
	
	Let $t' = t - t_0$ and $V = U^{\Dec}_{t'} \cdot U^{\Dec}_{t'-1} \cdots U^{\Dec}_1$, we have
	\[
	(\phi_{t})_{\eancC\{1\}} = V (\phi_{t_0})_{\eancC\{1\}} V^\dagger.
	\]
	
	Now, since $\val_\caI(\phi) = \val(I)$, it follows that
	\[
	\| (\phi_{t_0})_{\eancC\{1\}}- \Enc(\proj{1}) \|_1 \le \alpha,
	\]
	which further implies that
	\[
	\| (\phi_{t})_{\eancC\{1\}}- V \Enc(\proj{1}) V^\dagger \|_1 \le \alpha.
	\]
	From above, we can compute an $\alpha$-approximation to $(\phi_{t})_{S \cap \eancC\{1\}}$, which completes the proof.
	
\end{proof}

%

The following two lemmas can be proved in a similar way as~\cite[Lemma~4.9 and Lemma~3.5]{BroadbentG20}, we omit the proofs here since they are repetitive.

\begin{lemma}
	Let $k,\gamma \in \N$ be two constants, $n$  a power of $2$, $\caI = (n,n,\{S_i\}_{i \in [n]},\{C_i\}_{i \in [n]})$ a $(k,\gamma)$-$\QSAT$ instance, $\HamEnc = \HamEnc(\caI)$, and $\alpha = 1 - \val(\caI)$. Let $\spz{\phi}$ be an $n$-qubit pure state such that $\val_\caI(\phi) = \val(\caI)$, and let $\spz{\Phi} \in \Hclock \otimes \Hstate$ be the  history state for $\Venc$ given the $3 \cdot n\cdot N$-qubit witness $\phi^{\otp}$ that is defined according to~\eqref{eq:def-otp}. For every $t_1,t_2 \in \zeroTon{T}$ such that $t_1 \le t_2 \le t_1 + 3$, and for any subset $S \subseteq \clock \cup \state$ such that the following holds
	\begin{itemize}		
		\item $S$ has at most $2$ intersections per block with $\edata,\eotp, \emidx,\eancC$, and $\eT$.
	\end{itemize}
	
	Let $\spz{\phi_t} = \Venc_{[1,t]} \Enc(\spz{\phi})_{\edata} \otimes \spz{0^{|\ancilla|}}_{\ancilla}$, one can compute an $|S|$-qubit state $\rho$ such that
	\[
	\left\|\left( \sum_{t,t' \in \{t_1,\dotsc,t_2\} } \spz{\phi_t} \rpz{\phi_{t'}}  \right)_{S} - \rho \right\|_1 \le \alpha.
	\]
	in $\poly(n,2^{|S|})$ time.
\end{lemma}

\begin{lemma}\label{lemma:all-ham-term-sim}
	Let $k,\gamma \in \N$ be two constants, $n$  a power of $2$, $\caI = (n,n,\{S_i\}_{i \in [n]},\{C_i\}_{i \in [n]})$  a $(k,\gamma)$-$\QSAT$ instance, $\HamEnc = \HamEnc(\caI)$, and $\alpha = 1 - \val(\caI)$. Let $\spz{\phi}$ be an $n$-qubit pure state such that $\val_\caI(\phi) = \val(\caI)$, and let $\spz{\Phi} \in \Hclock \otimes \Hstate$  be the history state for $\Venc$ given the $3 \cdot n\cdot N$-qubit witness $\phi^{\otp}$ that is defined according to~\eqref{eq:def-otp}. For every term $J$ in $\HamEnc$, the following hold:
	\begin{itemize}
		\item If $J$ is not one of the $\Hamprop_{5,j}$ or the $\Hamprop_{6,j}$, $(\Phi)_{\supp(J)}$ can be computed in $\poly(n)$ time.
		\item If $J$ is one of the $\Hamprop_{5,j}$ or the $\Hamprop_{6,j}$, then $J = \sum_{i \in [n]} J_i \otimes \proj{\Enc(i)}_{\eidx}$. For every $i \in [n]$, $(\Phi)_{\supp(J_i) \cup \eidx}$ can be computed in $\poly(n)$ time.
	\end{itemize}
\end{lemma}

Now we are ready to prove~\autoref{theo:sim} (restated below).

\begin{reminder}{\autoref{theo:sim}}
	For every $k,\gamma \in \N$ and every $\alpha,\beta \colon \N \to [0,1]$ such that $\alpha$ is negligible, it holds that:
	\[
	\text{$(k,\gamma)$-$\History$-$\EncprojLH_{\alpha,\beta}$ $\subseteq$ $(O(\log n),2)$-$\SimQMA_{1-\beta(n)/M,1-\alpha(n)/M}$},
	\]
	where $M = \aM \cdot k + \bM \cdot \gamma + \cM$.
\end{reminder}
\begin{proof}
	Let $n$ be a power of $2$, $\caI = (n,n,\{S_i\}_{i \in [n]},\{C_i\}_{i \in [n]})$  a $(k,\gamma)$-$\QSAT$ instance, and $\HamEnc = \HamEnc(\caI)$  a $(k,\gamma)$-$\History$-$\EncprojLH_{\alpha,\beta}$ instance. 
	
	We let $\VHenc$ be the verifier. We will show $\VHenc$ is a locally simulable verifier for $(k,\gamma)$-$\History$-$\EncprojLH_{\alpha,\beta}$. We first show the completeness and soundness:
	\begin{enumerate}
		\item If $\val(\caI) \le \alpha(n)$, then by~\autoref{lemma:enc-local}, $\VHenc$ accepts with probability at least $1-\alpha(n)/M$.
		\item If $\val(\caI) \ge \beta(n)$, then again by~\autoref{lemma:enc-local}, $\VHenc$ accepts with probability at most $1 - \beta(n)/M$.
	\end{enumerate}
	
	Next, from its definition, $\VHenc$ is $2$-adaptive, and at each round, it queries $O(\log n)$ qubits. It is simulable from~\autoref{lemma:all-ham-term-sim} and the assumption that $\alpha$ is negligible.
\end{proof}

\subsection{Proof of~\autoref{theo:LocalQMA-in-SimQMA}}

Finally, we are ready to prove~\autoref{theo:LocalQMA-in-SimQMA} (restated below).

\begin{reminder}{\autoref{theo:LocalQMA-in-SimQMA}}
	For every $k,\gamma \in \N$, $0 < \beta < 1$, and negligible function $\alpha$, there are $s \in (0,1)$ and $c \colon \N \to [0,1]$ such that $1 - c(n) \le \negl(n)$ and the following holds
	\[
	(k,\gamma)\text{-}\LocalQMA_{\beta,1-\alpha(n)} \subseteq (O(\log n),2)\text{-}\SimQMA_{s,c(n)}.
	\]
	\vspace{-2em}
\end{reminder}
\begin{proof}
	Let $L = (L_{\yes}, L_{\no}) \in (k,\gamma)\text{-}\LocalQMA_{\alpha(n),\beta}$. Let $m,p$ be the polynomials in the definition of $\LocalQMA$ (see~\autoref{defi:localQMA}). By repeating a local check multiple times and adding dummy qubits, we can assume that $m = p$. For simplicity we will also use $m$ to denote $m(n)$.
	
	Given an input $x \in \bits^n$, we then construct a $(k,\gamma)$-$\QSAT$ instance $\caI = (m,m,\{S_i\}_{i \in [m]},\{C_i\}_{i \in [m]})$ such that (1) $x \in L_{\yes}$ implies that $\val(\caI) \ge 1 - \alpha(n)$ and (2) $x \in L_{\no}$ implies that $\val(\caI) \le \beta$.
	
	Now, let $m' = m'(n) = 2^{\lceil \log m(n) \rceil}$ be the smallest power of $2$ that is larger than $m$. We create another instance $\caI'$ by adding $m' - m$ dummy circuits $C_i$ that always outputs $1$ into $\caI$. Note that we have
	\[
	\val(\caI') = \frac{m'-m}{m'} + \frac{m}{m'} \cdot \val(\caI).
	\]
	
	Note that $m \le m' \le 2m$, it follows that (1) $x \in L_{\yes}$ implies that $\val(\caI') \ge 1 - \alpha(n)$ and (2) $x \in L_{\no}$ implies that $\val(\caI) \le 1/2 + \beta/2$. We let $\beta' = 1/2 + \beta/2$.
	
	
	Let $\HamEnc = \HamEnc(\caI')$. By~\autoref{lemma:HamEnc-analysis}, we further have that (1) $x \in L_{\yes}$ implies that $\lmin(\HamEnc) \le \alpha(n)$ and (2) $x \in L_{\yes}$ implies that $\lmin(\HamEnc) \ge (1-\beta')^2 / p_1(k,\gamma)^2$, where $p_1$ is the polynomial in~\autoref{lemma:HamEnc-analysis}.
	
	Now we set $\balpha$ so that $\balpha(m'(n)) = \alpha(n)$ and $\bbeta = (1-\beta')^2 / p_1(k,\gamma)^2$, which is a constant. Hence, we have obtained a polynomial-time reduction (since $m'$ is bounded by a polynomial) from $L$ to $(k,\gamma)$-$\History$-$\EncprojLH_{\balpha,\bbeta}$. The theorem then follows from~\autoref{theo:sim} and the fact that $\balpha$ is negligible since $\alpha$ is negligible and $m'$ is bounded by a polynomial.
\end{proof}

	\section{A Candidate Zero-Knowledge Succinct Argument for $\SimQMA$ in $\QHROM$}\label{sec:zk-suc-cand}



In this section, we present a candidate zero-knowledge succinct argument for $\SimQMA$ in $\QHROM$. Combining with~\autoref{theo:LocalQMA-in-SimQMA}, this also extends to all of $\LocalQMA$ with $1 - \negl(n)$ completeness. Assuming $\QPCP_{\negl}$, this further extends to all of $\QMA$. 

Our candidate construction is a simple adaption of the Quantum-Merkle-tree-based candidate succinct argument $\Pisuc$ for $\LocalQMA$ from~\cite{ChenM22}. The quantum Merkle tree in~\cite{ChenM22} only allows one round of queries. We will need a natural modification of it to simulate the local verifier for a language from $(O(\log n),2)\text{-}\SimQMA_{s,c(n)}$, which has $2$ rounds of queries.


\subsection{The Zero-Knowledge Succinct Protocol $\Pizk$}
\newcommand{\id}{\mathsf{id}}
\newcommand{\blk}{b}
\newcommand{\ellL}{\ell_{L}}

\paragraph{Notation.} 

Let $L = (L_{\yes}, L_{\no}) \in (k(n),\ellL(n))\text{-}\SimQMA_{s(n),c(n)}$ for locality and round parameters $k(n),\ellL(n) \in \N$ and soundness and completeness parameters $s(n),c(n) \in [0,1]$ such that $s(n) < c(n)$. We will always assume that $k(n) \cdot \ellL(n) \le O(\log n)$. Let $m_L$ and $p_L$ be the polynomials and $V_L$ be the $k(n)$-local verifier in~\autoref{defi:simQMA}. Throughout this section, we will always use $n$ to denote the length of an input to $L$, $N = p_L(n)$ to denote the number of qubits in a witness for $V_L$, and $\lambda$ to denote the security parameter. When the meaning is clear, we will often use $k$  to denote $k(n)$  for simplicity (similarly for $\ellL(n),s(n),$ and $c(n)$).

We also set $\blk = \lambda /3$, and $\ell = N$. We assume that $\blk$ is an integer and $\ell$ is a power of $2$ for simplicity and without loss of generality since one can always add dummy qubits to the witness.


\paragraph*{The perfect binary tree $T_\ell$.} We will consider a perfect binary tree $T_\ell$ of $\ell$ leafs (see~\autoref{fig:binarytree} for an illustration). Note that $T_\ell$ has $\log\ell$ layers. We label the nodes in $T_\ell$ first from root to leaves and then from left to right, starting with $1$.

For a node $u$ in $T_\ell$, we observe that $u$'s parent is $\lfloor u/2 \rfloor$ if $u$ is not the root (\ie, $u \ne 1$) and $u$'s two children are $2u$ and $2u + 1$ if $u$ is not a leaf (\ie, $u < \ell$). We use $P_u$ to denote the set of nodes consisting of $u$ and all ancestors of $u$. Formally, we have
\[
P_u = \begin{cases}
\{u\} \quad& \text{if $u = 1$,}\\
\{u\} \cup P_{\lfloor u/2 \rfloor} \quad& \text{if $u > 1$.}
\end{cases}
\]

We also define $R_u$ as follows:
\[
R_u = \{\text{$u\in P_u$ or $\lfloor u/2\rfloor \in P$}  : u \in [2\ell - 1] \}.
\]
That is, a node $v$ belongs to $R_u$ if either $v$ is in $P_u$ or the parent of $v$ is in $P_u$. Also, for a set of nodes $S \subseteq [2\ell - 1]$, we set $R_S = \bigcup_{u \in S} R_u$. To avoid nested subscripts, we sometimes  use $R(S)$ to denote $R_S$.

\begin{figure}
\begin{center}
\begin{tikzpicture}[level/.style={sibling distance=80mm/#1},scale = 0.8]
    \node [circle,draw]{$1$}
      child {
      	node [circle,draw] {$2$}
      	child {
      		node [circle,draw] {$4$}
      		child {
      			node {$\vdots$} 
      			child {node [circle,draw,scale = 0.75] {\footnotesize$\ell+0$}}
      			child {node [circle,draw,scale = 0.75] (a) {\footnotesize$\ell+1$}}
      		}
      		child {node {$\vdots$}}
      	}
      	child {
      		node [circle,draw] {$5$}
      		child {node {$\vdots$}}
      		child {node {$\vdots$}}
      	}
      }
      child {
      	node [circle,draw] {$3$}
      	child {
      		node [circle,draw] {$6$}
      		child {node {$\vdots$}}
      		child {node {$\vdots$}}
      	}
      	child {
      		node [circle,draw] {$7$}
      		child {node {$\vdots$}}
      		child {node {$\vdots$}
      			child {node [circle,draw,scale = 0.75] (b) {\footnotesize$2\ell-2$}}
      			child {node [circle,draw,scale = 0.75] {\footnotesize$2\ell-1$}}
      		}
      	}
      };
      \path (a) -- (b) node [midway] {$\cdots\cdots\cdots\cdots\cdots\cdots\cdots\cdots\cdots\cdots\cdots\cdots$};
\end{tikzpicture}
\caption{An illustration of the labeling of the nodes in the tree $T_\ell$ with $\ell$ leaves}\label{fig:binarytree}
\end{center}
\end{figure}
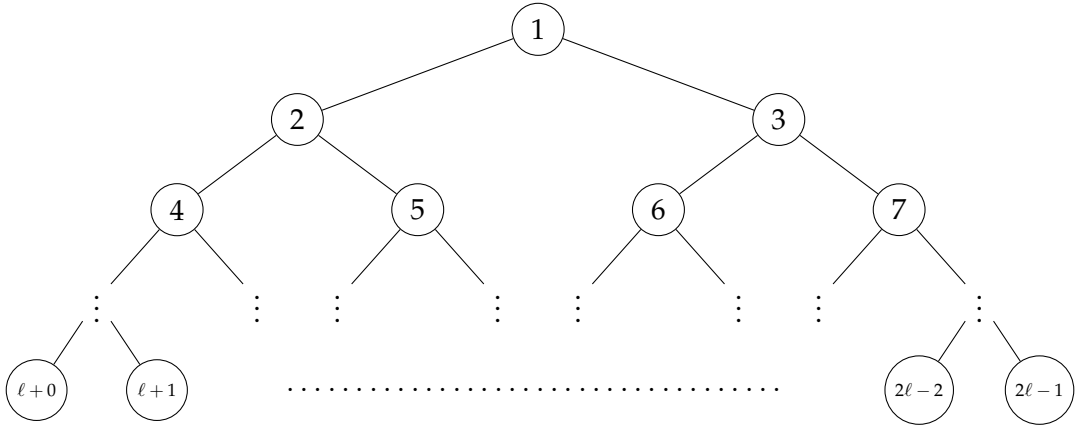

Given an $N$-qubit state $\sigma$, we first recall the following commitment algorithm (\autoref{algo:commit}) from~\cite{ChenM22}. Below we change the procedure slightly so that one leaf of the tree only stores a single qubit from the $N$-qubit quantum state.

\begin{algorithm}[H]
	\caption{Algorithm for committing to an $N$-qubit quantum state}\label{algo:commit}
	\SetKwProg{Fn}{Function}{}{}
	\Fn{$\commit^{\caG}(\sigma,N,\lambda)$}{
		\KwIn{$\sigma$ is an $N$-qubit quantum state; $\lambda$ is the security parameter (recall that $\lambda = 3\blk$)}
		Let $\ell = N$\; 
		For each node $u$ in $T_\ell$, create a $\blk$-qubit register $\state^{(u)}$\;
		Store $\sigma$ in  $\state^{(\ell)}(1),\state^{(\ell+1)}(1),\dotsc,\state^{(2\ell-1)}(1)$, other qubits of these registers are initialized to the all-zero state\; 
		\For{$u$ from $\ell - 1$ down to $1$}{
			Initialize $\state^{(u)}$ as $\spz{0^\blk}$\;
			Apply $\caG$ on $\state^{(2u)}$, $\state^{(2u+1)}$, and $\state^{(u)}$\;
		}
		\Return all qubits in $\{ \state^{(u)} \}_{u \in [2\ell - 1]}$\;
	}
\end{algorithm}

Next, we describe the following local decommitment algorithm (\autoref{algo:decommit}), which allows multiple rounds of queries.

\newcommand{\Sold}{S^{\mathsf{old}}}
\newcommand{\Snew}{S^{\mathsf{new}}}

\begin{algorithm}[H]
	\caption{Algorithm for recovering part of the original quantum state; it supports multiple rounds of queries }\label{algo:decommit}
	\SetKwProg{Fn}{Function}{}{}
	\Fn{$\decommit^{\caG^\dagger}\left(N,\lambda, \Sold, \Snew, \{ \eta_{u} \}_{u \in R(\Snew) \cup R(\Sold)}\right)$}{
		\KwIn{$\Sold,\Snew \subseteq \{\ell,\ell+1,\dotsc,2\ell - 1\}$ are two subsets of leaves in $T_\ell$, $\Sold$ denotes the subset that is already known and $\Snew$ denotes the subset we are going to decommit;
		$\newline$for each $u \in R(\Snew) \cup R(\Sold)$, $\eta_u$ is a $\blk$-qubit quantum state; $\lambda$ is the security parameter}
		Let $\ell = N$\; 
		For each node $u$ in $R(\Snew) \cup R(\Sold)$, create a $\blk$-qubit register $\state^{(u)}$, and store $\eta_u$ in $\state^{(u)}$\;
		\For{$u \in \Big(R(\Snew) \setminus R(\Sold)\Big) \cap [\ell - 1]$, from the smallest to the largest\label{line:order-decom}}{
			Apply $\caG^\dagger$ on $\state^{(2u)}$, $\state^{(2u+1)}$, and $\state^{(u)}$\;
			Measure $\state^{(u)}$ in the computational basis to obtain an outcome $z \in \bits^{\blk}$\;
			\lIf{$z \ne 0^{\blk}$}{\Return $\perp$}
		}
		\Return all qubits in $\{ \state^{(u)} \}_{u \in S}$\;
	}
\end{algorithm}

\newcommand{\Sent}{\mathsf{Sent}}

Finally, we are ready to specify the following candidate zero-knowledge succinct argument for $L \in (k,\ellL)\text{-}\SimQMA_{s,c}$.

\begin{construction}{The candidate zero-knowledge succinct argument $\Pizk$ for $L \in (k,\ellL)\text{-}\LocalQMA_{s,c}$}
	\begin{itemize}
		\item Both prover ($\caP$) and verifier ($\caV$) get access to a Haar random quantum unitary $\caG$ acting on $3\blk = \lambda$ qubits and its inverse $\caG^\dagger$. They also both get an input $x \in \bits^n$ to $L$. The goal for the prover is to convince the verifier that $x \in L_{\yes}$.
		
		
		Let $\ell = N$ and we assume that $\ell = 2^d$ for $d \in \N$.
		
		\item (\textbf{First message: $\caP \to \caV$}) The honest prover $\caP$ acts as follows: If $x \in L_{\no}$, $\caP$ aborts immediately. Otherwise, $\caP$ finds an $N$-qubit state $\sigma$ such that $V_L$ accepts with probability at least $c$, and runs $\commit(\sigma,N,\lambda)$ to obtain qubits $\{ \eta_{u} \}_{u \in [2\ell - 1]}$.
		
		$\caP$ then sends $\eta_{1}$ to $\caV$.
		
		\item (\textbf{Simulation of $V_L$}) $\caV$ now simulates the local verifier $V_L$. It first draws $\tau_0 \getsR [m_L(n)]$.
		
		Next for each $i \in [\ellL]$:
		
		\begin{enumerate}
			\item (\textbf{First turn: $\caV \to \caP$}) $\caV$ first simulates $V_L$ to compute a subset $S_i \subseteq [N]$, given previous outcomes $\tau_{\le i-1}$, together with a POVM $\{ \Pi_j \}_{j \in [m_L(n)]}$ on $|S_i|$ qubits. Note that $S_i \cap S_j = \emptyset$ for all $j < i$. It then sends $\tau_{i-1}$ to $\caP$.
			
			Let $W_i$ be the set of leaves in $T_\ell$ that contains the qubits indexed by $S_i$. That is,
			\[
			W_i = \{ \ell + u - 1 : u \in S_i \}.
			\]
			
			Let $\Sent_{\le i}$ to denote $\bigcup_{j \in [i]} W_j$. For simplicity, we  let $\Sent_{\le 0} = \Sent_0 = \{1\}$.
			
			\item (\textbf{Second turn: $\caP \to \caV$}) The honest prover $\caP$ sends the following to $\caV$   \footnote{Given $\tau_{i-1}$, $\caP$ would be able to compute the set $S_i$ and thus $W_i$. $\caP$ also already knows $\Sent_{\le i-1}$.} 
			\[
			\{\eta_u\}_{u \in R(W_i) \setminus R(\Sent_{\le i-1})}.
			\]

			\item (\textbf{Checking}) $\caV$ then runs $\decommit(N,\lambda,\Sent_{\le i-1},W_i,\{ \eta_u \}_{u \in R(\Sent_{\le i})})$. Note that $\caV$ already has $\{ \eta_u \}_{u \in \Sent_{\le i-1}}$. If $\decommit$ returns $\perp$, $\caV$ rejects immediately. Otherwise, $\caV$ continues the simulation of $V_L$ by measuring $\{ \Pi_j \}_{j \in [m_L(n)]}$ on the corresponding qubits from the $\eta_u$'s to obtain an outcome $\tau_i$.
		\end{enumerate}
	
		\item (\textbf{Final decision}) Finally, $\caV$ simulates $V_L$ based on the sequence $\tau_{\le \ellL}$ to decide whether it accepts or not.
		
	\end{itemize}
\end{construction}

\subsection{Analysis of $\Pizk$}

Now we prove the completeness and succinctness of $\Pizk$.

\begin{theorem}\label{theo:Pizk-c-and-s}
	Let $\Pizk$ be the protocol between $\caP$ and $\caV$ for the promise language $L \in (k,\ellL)\text{-}\SimQMA_{s,c}$. For every $x \in \bits^n$, the following hold:
	\begin{description}
		\item[] \textbf{Completeness:}
		{ 	If $x \in L_{\sf yes}$, then for every $\caG \in \Haar(2^\lambda)$,
			\[
			\Pr[(\caV^{\caG,\caG^\dagger} \leftrightarrows \caP^{\caG,\caG^\dagger} )(x) = 1] \ge c.
			\]}
		\item[]  \textbf{Succinctness:} {
			$\caP$ and $\caV$ communicate at most $O(k \cdot \ellL \cdot \log n \cdot \lambda)$ qubits in total. }
		
		\item[] \textbf{Efficiency:} {$\caV$ runs in $\poly(n,k,\gamma)$ time. If $V_L$ is strongly explicit, then $\caV$ runs in $O(k \cdot \ellL \cdot \log n \cdot \lambda + \poly(\log n,k,\gamma))$ time.
		}
	\end{description}
\end{theorem}
\begin{proof}
	We first establish the succinctness part. Examining the protocol $\Pizk$, one can see that the first message takes $O(\lambda)$ qubits. For the later $\ellL$ rounds, the message of the first turn takes $O(\log m_L(n)) = O(\log n)$ classical bits, and the message of the second turn takes at most $O\left( |R_{W_i}| \cdot \lambda \right)$ qubits. Note that $|R_{W_i}| \le |W_i| \cdot O(\log \ell) \le k \cdot O(\log N) \le O(k \cdot \log n)$, the total communication complexity is thus bounded by $O(k \cdot \ellL \cdot \log n \cdot \lambda)$. 
	
	For the running time of $\caV$, one can see that its running time is dominated by the running time of $\decommit$ and the running time of $V_L$ computing $W_i$ and $\{ \Pi_j \}$, which are at most $O(k \cdot \log N \cdot \lambda)$ and $\poly(n,k,\gamma)$ ($\poly(\log n,k,\gamma)$ if $V_L$ is strongly explicit), respectively.
	
	Now we prove the completeness. Let $\caG_{(u)}$ be a $\caG$ gate applying on registers $\state^{(2u)}$, $\state^{(2u+1)}$, and $\state^{(u)}$. Then we know for the honest prover $\caP$, when $x \in L_{\yes}$, it prepares an $N$-qubit state $\sigma$ that makes $V_L$ accepts with probability at least $c$, and then applies $U_{\sf com} \coloneqq \caG_{(1)} \cdot \dotsc\cdot\caG_{(\ell-1)}$ to $\sigma \otimes \proj{0}_{\state^{(1)},\dotsc,\state^{(\ell-1)}}$.
	
	Let $U_{\sf decom} \coloneqq U_{\sf com}^{\dagger} = \caG_{(\ell-1)}^\dagger\cdot \dotsc \cdot \caG_{(1)}^\dagger$. Recall that verifier $\caV$ at the $i$-th round simulates $V_L$ only on registers in $\{ \state^{(u)} \}_{u \in W_i}$. We now argue that $\caV$ is effectively simulating $V_L$ on
	\[
	U_{\sf decom}^{\dagger} U_{\sf com} \sigma \otimes \proj{0}_{\state^{(1)},\dotsc,\state^{(\ell-1)}} = \sigma \otimes \proj{0}_{\state^{(1)},\dotsc,\state^{(\ell-1)}}.
	\]
	
	First, we can see that when $i = 1$ and $\Sent_{\le i-1} = \{ 1 \}$, $\decommit(N,\lambda,\Sent_{\le i-1},W_i,\{ \eta_u \}_{u \in R(\Sent_{\le i})})$ performs all gates in $U_{\sf decom}$ that lies in the lightcone of the registers $\{ \state^{(u)} \}_{u \in W_i}$ in the chronological order (see Line~\ref{line:order-decom} of \autoref{algo:decommit}). Also, since $\caP$ starts with the state $\sigma \otimes \proj{0}_{\state^{(1)},\dotsc,\state^{(\ell-1)}}$, $\decommit$ never outputs $\bot$. For $i > 1$, we also observe that $\decommit(N,\lambda,\Sent_{\le i-1},W_i,\{ \eta_u \}_{u \in R(\Sent_{\le i})})$ performs the additional gates in $U_{\sf decom} $ in the correct order so that all gates in $U_{\sf decom}$ that lie in the lightcone of the registers $\{ \state^{(u)} \}_{u \in \Sent_{\le i}}$ are performed in the chronological order. Therefore, $\caV$ is simulating $V_L$ faithfully on $\sigma$, meaning that it accepts with probability at least $c$.
\end{proof}

Similarly to~\cite{ChenM22}, we conjecture that the soundness also holds.
\begin{conjecture}[$\Pizk$ is sound in $\QHROM$]\label{conj:soundness}
	Let $\Pizk$ be the protocol between $\caP$ and $\caV$ for the promise language $L \in (k,\ellL)\text{-}\SimQMA_{s,c}$. For every $x \in \bits^n$, the following hold:
	\begin{description}
		\item[] \textbf{Soundness:}
		{ 
			If $x \in L_{\no}$, then for every $t \in \N$ and all (potentially malicious) $\caP^*$ that make at most $t$ total queries to $\caG$ and $\caG^\dagger$, for some $\delta = \delta(t,\lambda) = \poly(t)/2^{\Omega(\lambda)}$, it holds that
			\[
			\Pr_{\caG \getsR \Haar(2^\lambda)} \left[ \Pr[(\caV^{\caG,\caG^\dagger}  \leftrightarrows (\caP^*)^{\caG,\caG^\dagger} ) (x) = 1)] \ge s + \delta \right] \le \delta.
			\]
		}
	\end{description}
\end{conjecture}

Finally, we conjecture that $\Pizk$ is zero knowledge.

\begin{conjecture}\label{conj:Pizk-is-zk}
	$\Pizk$ is computational zero knowledge in the $\QHROM$.
\end{conjecture}

Similarly to~\cite{ChenM22}, we remark that (1) the constant soundness in \autoref{conj:soundness} and the constant completeness in~\autoref{theo:Pizk-c-and-s} can be easily amplified to $n^{-\omega(1)}$ and $1-n^{-\omega(1)}$ by repeating the protocols $\log^2 n$ times, and (2) assuming $\QPCP$, the protocol works for all languages in $\QMA$.

\begin{corollary}\label{cor:suc-protocol-for-SimQMA}
	Assuming \autoref{conj:soundness} and~\autoref{conj:Pizk-is-zk}, there is a computational zero-knowledge protocol for $L \in (k,\ellL)\text{-}\SimQMA_{s,c}$ with $\lambda \cdot \polylog(n)$ communication complexity, completeness $1-n^{-\omega(1)}$ and soundness $n^{-\omega(1)}$ in $\QHROM$. Also, if $V_L$ is strongly explicit, then the verifier running time of the protocol is also bounded by $\lambda \cdot \polylog(n)$.
	
	Moreover, if we further assume that $\QPCP$ holds, then the aforementioned succinct zero-knowledge protocol exists for every $L \in \QMA$.
\end{corollary}

	\section*{Acknowledgments} L.C. would like to thank Jiahui Liu and Qipeng Liu for helpful discussions and pointing out many related works. This work was done while L.C. did an internship at IBM Quantum Research.
	
	\bibliographystyle{alpha}
	\bibliography{literature}
	
	\appendix
	
\end{document}